\documentclass[11pt]{article}
\usepackage[margin = 1in]{geometry}

\usepackage[numbers]{natbib}

\newcount\Comments  

\newcommand{\kibitz}[2]{\ifnum\Comments=1{\color{#1}{#2}}\fi}

\newcommand{\hma}[1]{\kibitz{blue}{[hma: #1]}}

\newcommand{\rmr}[1]{\kibitz{auburn}{[rm: #1]}}

\newcount\vsqCounter
\newcommand{\vsq}[1]{\ifnum\vsqCounter=1{\vspace{#1}}\fi}

\usepackage{etex}
\usepackage{latexsym}
\usepackage{amssymb}
\usepackage{amsmath}
\usepackage{amsfonts}
\usepackage{bbm}
\usepackage{ifthen}
\usepackage{enumerate}
\usepackage{psfrag}
\usepackage{floatflt}
\usepackage{verbatim}
\usepackage{dsfont}
\usepackage{scalefnt}
\usepackage{calc}
\usepackage{graphicx}
\usepackage{subcaption}
\usepackage{multicol}
\usepackage{tikz}
\usetikzlibrary{shapes.multipart,arrows,automata}
\usetikzlibrary{patterns,calc,intersections,through,backgrounds,decorations.pathreplacing}
\usepackage{xstring}

\usepackage{amsthm}		
\usepackage{thmtools, thm-restate}

\usepackage[linesnumbered,ruled,vlined]{algorithm2e}

\usepackage{color, xcolor} 
\colorlet{darkblue}{blue!40!black}
\definecolor{auburn}{rgb}{0.43, 0.21, 0.1}
\definecolor{orange}{rgb}{1, 0.5, 0}
\definecolor{lightblue}{rgb}{0.1176, 0.5647, 1}

\definecolor{matlabblue}{rgb}{0    0.4470    0.7410}
\definecolor{matlabred}{rgb}{0.8500    0.3250    0.0980}
\definecolor{matlabyellow}{rgb}{ 0.9290    0.6940    0.1250}
\definecolor{matlabpurple}{rgb}{0.4940    0.1840    0.5560}

\definecolor{auburn}{rgb}{0.43, 0.21, 0.1}
\colorlet{darkblue}{blue!35!black}

\usepackage{hyperref}
\hypersetup{
	colorlinks = true,
	citecolor = darkblue,
	urlcolor = darkblue,
	linkcolor = darkblue
}

\theoremstyle{plain}

\newtheorem{proposition}{Proposition}

\theoremstyle{definition}
\newtheorem{definition}{Definition} 
\newtheorem{example}{Example}

\newcommand{\report}{r}					
\newcommand{\reportSet}{\mathcal{R}}	

\newcommand{\betahat}{\tilde{\beta}}		
\newcommand{\thetahat}{\tilde{\theta}}

\newcommand{\mech}{\mathcal{M}}			
\newcommand{\tzero}{s}					
\newcommand{\tone}{t}						

\newcommand{\uhat}{\tilde{u}}				
\newcommand{\Uhat}{\tilde{U}}

\newcommand{\zc}{z^0}
\newcommand{\zcmax}{\varphi}		
\newcommand{\zcmaxhat}{\tilde{\zcmax}} 

\newcommand{\ut}{ut}					
\newcommand{\sw}{sw}					
\newcommand{\rev}{\mathit{rev}}		

\newcommand{\bmin}{\underline{b}}
\newcommand{\bmax}{\bar{b}}

\newcommand{\fixedV}{w}				
\newcommand{\fixedP}{p}				

\newcommand{\0}{^{(0)}}
\newcommand{\1}{^{(1)}}
\newcommand{\2}{^{(2)}}
\renewcommand{\th}{^{\mathrm{th}}}
\newcommand{\fb}{^{\mathrm{FB}}}

\newcommand{\setR}{\mathbb{R}}

\newcommand{\txtwp}{~\mathrm{w.p.}~}
\newcommand{\txtif}{~\mathrm{if}~}

\providecommand{\pwfun}[1]{\left\lbrace \begin{array}{ll} #1 \end{array} \right.}
\newcommand{\one}[1]{\mathds{1} \{ #1\}}
\newcommand{\E}[1]{\mathbb{E}\left[ #1 \right]}
\newcommand{\Pm}[1]{\mathbb{P}\left[ #1 \right]}

\newcommand{\txtSP}{{\mathrm{SP}}}

\newcommand{\txtCSP}{{\mathrm{CSP}}}

\title{Penalty Bidding Mechanisms for Allocating Resources and Overcoming Present Bias%
\thanks{The authors would like to thank Yiling Chen, Ido Erev, Matt Juszczak, Scott Kominers, Jake Marcinek, 
and Kyle Pasake for helpful comments and discussions.}
}

\author{Hongyao Ma%
\thanks{John A. Paulson School of Engineering and Applied Sciences, Harvard University, Cambridge,
MA, 02138, USA. Email: hongyaoma@seas.harvard.edu.} 
\and 
Reshef Meir%
\thanks{Department of Industrial Engineering and Management,
Technion - Israel Institute of Technology, Technion City, Haifa 3200003, Israel. Email: reshefm@ie.technion.ac.il.}
\and
David C. Parkes%
\thanks{John A. Paulson School of Engineering and Applied Sciences, Harvard University, 33 Oxford Street, Maxwell Dworkin 229, Cambridge, MA 02138, USA. Email: parkes@eecs.harvard.edu.}
\and Elena Wu-Yan%
\thanks{Harvard College, Cambridge,
MA, 02138, USA. Email: elenaw@college.harvard.edu.} 
}

\begin{document}

\maketitle

\begin{abstract}
From skipped exercise classes to last-minute cancellation of dentist appointments, underutilization of reserved resources abounds. Likely reasons include uncertainty about the future, further exacerbated by present bias. In this paper, we unite resource allocation and commitment devices through the design of contingent payment mechanisms, and propose the {\em two-bid penalty-bidding mechanism}. This extends an earlier mechanism proposed by \citet{ma2019contingent}, assigning the resources based on willingness to accept a no-show penalty, while also allowing each participant to increase her own penalty in order to counter present bias. We establish a simple dominant strategy equilibrium, regardless of an agent's level of present bias or degree of ``sophistication''. Via simulations, we show that the proposed mechanism substantially improves utilization and achieves higher welfare and better equity in comparison with mechanisms used in practice and mechanisms that optimize welfare in the absence of present bias. 
\end{abstract}

\section{Introduction} \label{sec:intro}


``It was a disaster,'' recalled Matt Juszczak, co-founder of Turnstyle Cycle and Bootcamp, a fitness company that offers cycling and bootcamp classes across five studios in the Boston area. ``When we opened our first indoor cycling location in Boston's Back Bay, we saw 40 to 50 no-shows and late cancels in an average day--- that's over 15,000 in a year!''\footnote{\href{https://business.mindbody.io/education/blog/tips-reduce-no-shows-and-late-cancels-your-fitness-business}{https://business.mindbody.io/education/blog/tips-reduce-no-shows-and-late-cancels-your-fitness-business}, visited September 1, 2018.}
Like many well-known exercise studios, Turnstyle allowed customers to reserve class spots several days in advance with a first-come-first-serve reservation system. However, ambitious customers,  overestimating the amount of time in their schedules or their desire to exercise in the future, often snag a spot only to ultimately cancel last-minute or simply not show up. 

\begin{figure}[t!]
\centering
\includegraphics[scale=0.24]{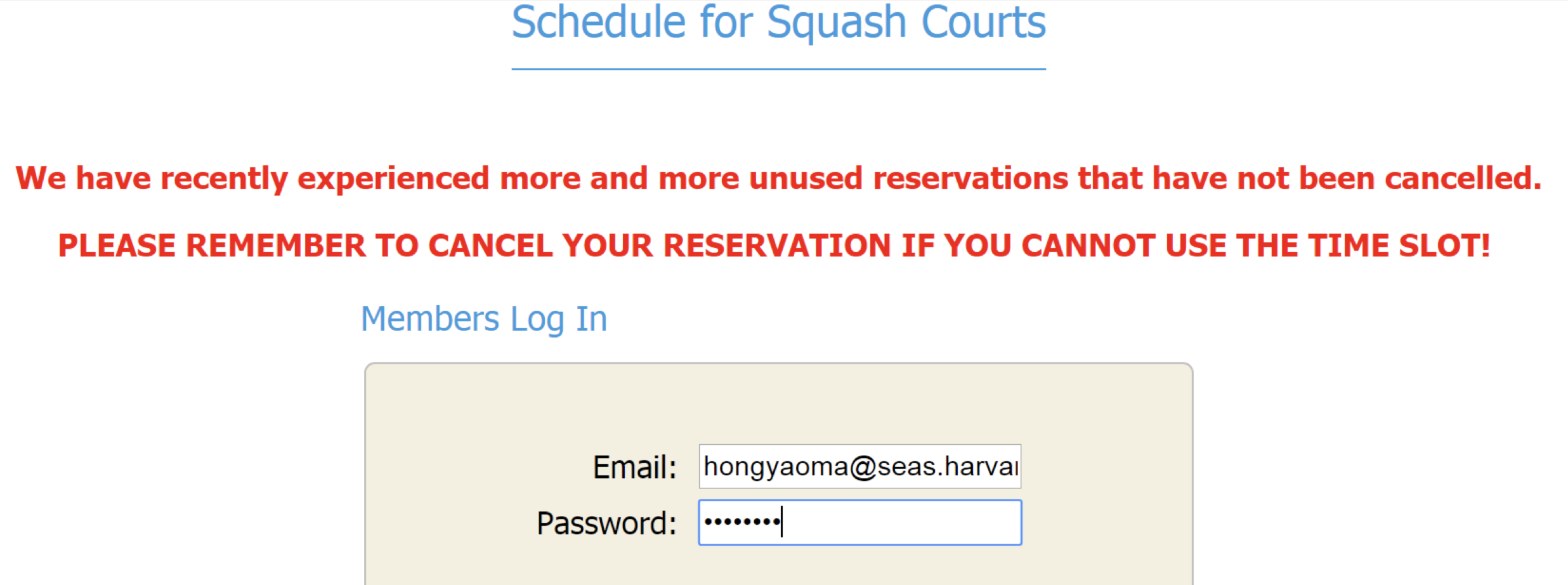}
\caption{Log in page of the squash court reservation system at the Harvard Hemenway Gymnasium. \vspace{-1em}  \label{fig:hemenway} 
}
\end{figure}

Similarly, the squash courts at Harvard's Hemenway Gymnasium used to allow members to reserve time-slots to play squash up to seven days ahead of time. Even though the reservation window has since been reduced to three days, the gym operators still feel the need to display the warning shown in Figure~\ref{fig:hemenway} every time someone logs into the reservation system.\footnote{\url{https://recreation.gocrimson.com/recreation/facilities/Hemenway}, visited September 1, 2018.}
For other examples, organizers of free events report to Eventbrite that their no-show rate can be as high as 50\%,\footnote{\url{https://www.eventbrite.com/blog/asset/ultimate-way-reduce-no-shows-free-events/}, visited 5/6/2019.} 
and even for prepaid events organized through Doorkeeper,
the fraction of no-shows can be  20\%.\footnote{\url{https://www.doorkeeper.jp/event-planning/increasing-participants-decreasing-no-shows?locale=en}, visited May 6, 2019.} 
Studies of outpatient clinics report that
no-shows can range from 23-34\%, with  no-shows costing
an estimated 14\% of daily revenue as well as impacting efficiency ~\cite{mehra2018reducing}. 

Common to all these examples is the presence of uncertainty, self-interest, and down-stream decisions by participants, together with the interest of the planner (gym manager, event organizer, health clinic) in a resource being used and not wasted. Beyond revenue and efficiency motivations, utilization can have positive externality in and of itself, cycling studio members gaining motivation from fellow bikers, for example. 
%
Complicating the problem is \emph{present bias}, often phrased as the constant struggle between our current and future selves~\citep{laibson1997golden,o1999doing}. It is easy to imagine that at the beginning of the week, someone might prefer a spin class over watching TV on Friday, reserving a spot, but by the time Friday comes around preferring to just watch TV.

Recognizing the problem of low utilization, many reservation systems charge a penalty for no-show. Turnstyle has started to charge a \$20 penalty for missing a class,\footnote{\url{https://kb.turnstylecycle.com/policies/what-is-the-late-cancel-no-show-policy}, visited May 6, 2019.}
patients who miss appointments at hospitals may need to pay a fee that is not covered by insurance,%
\footnote{\url{https://huhs.harvard.edu/sites/default/files/HDS\%20New\%20Patient\%20Welcome\%20Letter-eps-converted-to.pdf}, visited May 10th, 2018.} 
and organizers of some conferences collect a deposit that is returned only to students who actually attended talks.\footnote{\url{https://risingstarsasia2018.ust.hk/guidelines.php}, visited May 10th, 2018.} These approaches can be viewed as ad-hoc, first-come-first serve schemes, for some choice of no-show penalty: a penalty that is too small is not effective, whereas a penalty that is too big will drive away participation in the scheme. 

%
In recent work, \citet{ma2019contingent} model participants' future value from using a resource as a random variable, and propose the
\emph{contingent second price mechanism} (CSP). The mechanism elicits from each participant a bid on the highest no-show penalty she is willing accept, assigns the resources to the highest bidders, and charges the highest losing bid as the penalty. The bids provide a good signal for participants' reliability, and the CSP mechanism provably optimizes utilization in dominant strategy equilibrium among a large family of mechanisms. With present bias, however, charging the highest losing bid as penalty no longer guarantees truthfulness: a rational 
participant always prefers smaller penalties, but a present-biased participant may favor larger penalties when a stronger commitment device is more effective in overcoming myopia.
%

\subsection{Our results}

In this paper, we unite through contingent payment mechanisms the allocation of scarce resources under uncertainty, and the design of commitment devices--- techniques that aim to overcome present bias and to fulfill a plan for desired future behavior.

We generalize the model proposed in \citet{ma2019contingent}, decomposing an agent's value for a resource into the {\em immediate value} and the {\em future value}. The immediate value is a random variable, and the value experienced at the time of using the resource (modeling for example the opportunity cost and present pain of going to the gym). The future value is not gained until some future time (consider, for example the future benefit from better health). 
We incorporate the standard quasi-hyperbolic discounting model for time-inconsistent preferences~\cite{laibson1997golden,o1999doing}, such that when an agent is making a decision on whether to use a resource, the future value is discounted by a present bias factor. 
Agents may also have different levels of sophistication in regard to their level of self-awareness, modeled by agents' belief on their own present bias factor--- a {\em naive agent} believes she does not discount the future, a {\em sophisticated agent} knows her bias factor precisely and is able to perfectly forecast her future actions, and a {\em partially naive agent} resides somewhere in between~\cite{o1999doing,o2001choice}.

In period~0, an agent's private information is the distribution of the immediate value, the (fixed) future value, and what she believes to be her present bias factor. 
A mechanism elicits information from each agent, assigns each of $m\geq 1$ resources, and may determine both a {\em base payment} that an assigned agent always pays, as well as a {\em penalty} for each assigned agent in the event of a no-show. In period~1, each assigned agent learns her immediate value, and with knowledge of the penalty and future value, decides (under the influence of present bias) whether or not to use the resource.

The \emph{two-bid penalty-bidding mechanism} (2BPB) works as follows. In period~0, the mechanism elicits a bid from each agent, representing the highest penalty she is willing to accept for no show, and assigns the resources to the $m$ highest bidders. 
To address the non-monotonicity of agent's expected utility in the penalty, the mechanism asks each assigned agent to report a penalty weakly higher than the $m+1\th$ bid, representing the actual amount she would like to be charged in the case of a no-show (thereby operating also as a commitment device).

Given the option to choose an optimal level of commitment weakly above the highest losing bid, it is a dominant strategy under the 2BPB mechanism for each agent to bid her maximum acceptable no-show penalty, regardless of her immediate value distribution, future value, level of present bias, or degree of sophistication (Theorem~\ref{thm:dse_two_bid_mech}). %
While naive agents do not see the value of commitment and generally do not take any commitment device when offered~\cite{bryan2010commitment,beshears2011self}, the 2BPB mechanism is still able to help  reducing the loss of welfare and utilization due to no show, since a commitment device is designed through the mechanism, and is an integral part of the system. 
We also prove that the mechanism satisfies voluntary participation, and runs without a budget deficit.


We show via simulation that the 2BPB mechanism not only improves utilization, but also achieves higher social welfare than the standard $m+1\th$ price auction, which is welfare-optimal for settings without present bias. 
The mechanism also outperforms a family of mechanisms widely used in practice, which assign resources first-come-first-served  and charge a fixed no-show penalty. Moreover, in a population where agents have different levels of present bias, the more biased agents benefit more than the less biased agents under the 2BPB mechanism. This results in better equity compared with the outcome under the $m+1\th$ price auction, where the most biased agents gain little or no welfare.

\subsection{Related Work} \label{sec:related_work}

To the best of our knowledge, this current paper\rmr{``this paper" or ``the current paper"} is the first to study resource assignment in the presence of uncertainty and
present bias. 
The closest related work is on the design of mechanisms to improve
resource utilization where agents have uncertain future
values~\cite{ma2019contingent,Ma_ijcai16,Ma_aamas17}.
%
The proposed mechanisms, however, no longer have dominant strategy equilibrium for present-biased agents.
This present work builds on \citet{ma2019contingent}, generalizing the model to incorporate present bias, and makes use of two-bid penalty bidding to align incentives.
Crucially, the 
mechanism does not need any knowledge about agents' level of bias or value distributions.

Contingent payments have arisen
in the past in the context of oil drilling license
auctions~\cite{hendricks1988empirical},
royalties~\cite{caves2003contracts,deb2014implementation}, ad
auctions~\cite{varian2007position}, and selling a
firm~\cite{ekmekci2016just}. Payments that are contingent on some
observable world state also play the role of improving revenue as well as  hedging risk~\cite{skrzypacz2013auctions}. In our model, in contrast, payments are contingent on agents' own downstream decisions and serve the role of commitment devices. 
In regard to auctions in which actions take place after the time of contracting, \citet{atakan2014auctions} study auctions where the value of taking each action depends on the collective actions by others, but these actions are taken before rather than after observing the world state. \citet{courty2000sequential} study the problem of revenue maximization in selling airline tickets, where passengers have uncertainty about their value for a trip, and may decide not to take a trip after realizing their actual values. The type space considered there is effectively one-dimensional, and present bias is not considered.


\citet{laibson1997golden} introduced the quasi-hyperbolic discounting for modeling time-inconsistent decision making, where in addition to exponential discounting, all future utilities are discounted by an additional present bias factor. 
%
%
\rmr{this can be slightly more connected to the intro text. E.g. "the distinction we presented above between naive and sophisticated present biased agents, is due to ...,  who also ..."}%
\citet{o1999doing} classify present-biased agents into naive agents
(unaware of present bias)
and sophisticated agents (fully aware), and find  that naive agent procrastinate immediate-cost activities and  
do immediate-reward activities too soon, while sophistication lessens procrastination but intensifies the
doing-too-soon. \citet{o2001choice} also study how the role of choice affects procrastination, and introduce the idea of a
partially naive agent, who is aware of present bias but underestimates the degree of this bias. 
Researchers have also attempted to estimate the present bias factor in the real world, 
however, there has not been consensus about this~\cite{augenblick2015experiment,cohen2016measuring,ericson2018intertemporal}.

%
Researchers have also examined various kinds of {\em commitment devices} to mitigate present bias.
\citet{gine2010put}, for example, offer smokers a savings account that forfeits deposits to a charity if the they fail a urine test for nicotine.
By bundling a ``want'' activity (listening to one's favorite audio book) with a ``should'' activity (going to the gym), \citet{milkman2013holding} evaluate the effectiveness of temptation bundling as a commitment device to tackle two self-control problems at a time. See also \citet{laibson1997golden} and \citet{beshears2015self}.
%
%
In a different setting, \citet{kleinberg2014time} consider how to modify the sequencing of tasks available to individuals in order to help a present-biased agent adopt a more optimal sequence of tasks.
This work is later extended to consider sophisticated agents, the interaction between present bias and sunk-cost bias, and agents whose present bias factors are uncertain~\cite{kleinberg2016planning,kleinberg2017planning,gravin2016procrastination}. 
There are no uncertain values or costs in these models, 
and no contention for limited resources.

\section{Preliminaries} \label{sec:preliminaries}

We first introduce the model for the assignment of $m$ homogeneous resources, leaving a discussion of the generalization to heterogeneous resources to Section~\ref{sec:discussion}. 
There is a set of agents $N = \{1,~\dots, ~n \}$ and three time periods. In period~$0$, when resources need to be assigned, 
the value of each agent $i \in N$ for using a resource is uncertain, represented by $V_i = V_i\1 + v_i\2$. 
The period~1 \emph{immediate value} 
$V_i\1$ is a random variable with cumulative distribution function  $F_i$, whose exact (and potentially negative) value is not realized
until period $1$. 
This models, for example, the opportunity cost and present pain of going to the gym. The period~$2$ {\em future value} $v_i\2 \geq 0$ models the expected future benefit for agent $i$ (
e.g. the future benefit from better health), if she uses a resource in period~$1$.

Agents are present-biased, such that at any point of time, when making decisions, agent $i$ discounts her utility from all future periods by a factor of $\beta_i \in [0,~1]$. Agents may not be fully aware of this bias, however, and agent $i$ believes that when making decisions, her future utility will be discounted by a factor of $\betahat_i \in [\beta_i,~1]$. 
An agent with $\betahat_i = \beta_ 1 = 1$ is \emph{rational} and does not discount her future utility. 
An agent with $\betahat_i = \beta_i < 1$ is said to be \emph{sophisticated}, and fully aware of the degree of her present bias (thus is able to correctly predict her future decisions). An agent with $\beta_i<1$ and $\betahat_i = 1$ is said to be \emph{naive}, believing that she will make rational decisions in the future, and an agent with $\betahat_i \in (\beta_i, 1)$ is said to be \emph{partially naive}.

Let $\theta_i = (F_i, v_i\2, \beta_i, \betahat_i)$ denote  agent $i$'s \emph{type}, and $\theta = (\theta_1, \dots, \theta_n)$ denote a type profile. The tuple $\thetahat_i = (F_i, v_i\2, \betahat_i)$ is agent $i$'s private information at period $0$, when the assignment of resources is determined. Each allocated agent privately learns the realization $v_i\1$ of $V_i\1$ and then decides whether to use the resource at period $1$. Define $V^+_i \triangleq \max \{V_{i}, \; 0\}$. Following \citet{ma2019contingent}, we make the following assumptions about $V_i$ for each $i \in N$:

\vsq{-0.5em}

\begin{enumerate}[({A}1)]
	\item $\E{V_{i}^+}>0$, which means that a rational agent gets positive value from using the resource with non-zero probability, thus the \emph{option} to use the resource 
has positive value. 
	\item $\E{V_{i}^+} < +\infty$, which means that agents do not get  infinite expected utility from the option to use the resource, thus would not be willing to pay an unboundedly large payment for it. 
	\item $\E{V_i} < 0$, meaning that being forced to always use the resource regardless of what happens is not favorable for any agent, so that no agent would accept any unboundedly large no-show penalty for the right to use a resource.\footnote{ 
Regardless of the degree of present bias or sophistication, an agent for which (A3) is violated is willing to accept a 1 billion dollar no-show penalty, (almost) always  use the resource, and get a non-negative utility in expectation.}
\end{enumerate}

We now provide a few examples of different models for agent types. 

%

\begin{example}[$(c_i,\fixedP_i)$ model] \label{ex:vipi} The future value for agent $i$ for using a resource is $v_i\2 = w_i > 0$, however, she is able to do so only with probability $\fixedP_i \in (0,1)$, and at a period~$1$ opportunity cost modeled by $V_i\1 = -c_i$. With probability $1 - \fixedP_i$, agent $i$ is unable to show up to use the resource. 
This hard constraint can be modeled as $V_i\1$ taking value $-\infty$ with probability $1-\fixedP_i$.
See Figure~\ref{fig:pmf_vipi}. We have $\E{V_i^+} = (w_i - c_i) \fixedP_i > 0$, and $\E{V_i} = -\infty$ thus assumptions (A1)-(A3) are satisfied. 
\end{example}

\begin{figure}[t!]
\centering
\begin{subfigure}[t]{0.45\textwidth}
	\centering
	\small{
	\begin{tikzpicture}[scale = 0.9][font=\normalsize]
		\draw (0,0) node  {$V_i\1 = \pwfun{-c_i, &\txtwp\ \fixedP_i \\
			-\infty, &\txtwp\ 1 - \fixedP_i}$};		
		\draw (0,-0.5) node{{\color{white} some text}};
	\end{tikzpicture}	
}
\end{subfigure}%
\hspace{2 em}
\begin{subfigure}[t]{0.45\textwidth}
\begin{tikzpicture}[scale = 1][font = \small]
\draw[->] 	(0.5,0) -- (2.8,0) node[anchor=north] {$v$};
\draw[->] 	(2,-0.2) -- (2,1);
\draw (2,0.8)  node[anchor=west] {$f_i(v)$};
\draw[dotted] (0,0) -- (0.5,0);
\draw (-0.5, 0) -- (0, 0);

\draw (1, 0) -- (1,0.6);
\draw [fill] (1, 0.6) circle [radius=0.04] node[anchor=south] {\small{$\fixedP_i$}};

\draw (-0.25, 0) -- (-0.25,0.3);
\draw [fill] (-0.25, 0.3) circle [radius=0.04]node[anchor=south] {\small{$1-\fixedP_i$}};

\draw	(1, 0) node[anchor= north] {\small{$-c_i$}}
		(-0.25, 0) node[anchor = north] {\small{$-\infty$}};

\end{tikzpicture}
\end{subfigure}%
\caption{Distribution of $V_i\1$ under the $(c_i,~ \fixedP_i)$ type model. \label{fig:pmf_vipi}} \vsq{-0.5em}
\end{figure}
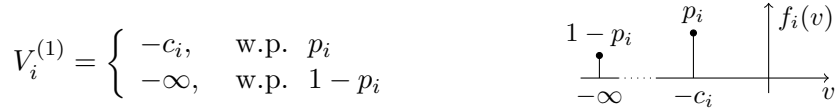

\begin{example}[Exponential model] \label{ex:exp_model} 
The opportunity cost for an agent to use the resource in period one is an exponentially distributed random variable with parameter $\lambda_i$, (i.e. $-V_i\1 \sim \mathrm{Exp}(\lambda_i)$). If the agent used a resource, she gains a future utility of $v_i\2 = \fixedV_i > 0$. 
See Figure~\ref{fig:pdf_exp}. The expectation of $\E{V_i\1}$ is $\lambda_i^{-1}$, thus $\E{V_i} = \lambda_i^{-1}+ \fixedV_i $ and (A1)-(A3) are satisfied when $w_i < \lambda_i^{-1}$.
\end{example}

\begin{figure}[t!]
\centering
\begin{subfigure}[t]{0.45\textwidth}
\centering
\small{
\begin{tikzpicture}[scale = 0.9][font=\normalsize]
	\draw (0,0) node  {$	f_i(v) = \pwfun{ \lambda_i e^{\lambda_i v}, & v \leq 0, \\
	0, & v > 0.}$};		
	\draw (0,-0.6) node{{\color{white} some text}};
\end{tikzpicture}	
}
\end{subfigure}%
%
%
%
\begin{subfigure}[t]{0.45\textwidth}
\centering
\begin{tikzpicture}[scale = 1.3][font=\small]

\draw[->] (-0.8,0.3) -- (2.2,0.3)  node[anchor=north] {$v$};
  
\draw[->] (1.5,0.1) -- (1.5,1.2);
\draw (1.5,1.1) node[anchor=west] {$f_i(v)$};

\draw[scale=1,domain=-0.5:1.5,smooth,variable=\x] plot (\x,{0.5*exp(1.1*(\x-1.5))+0.3});

\draw[-] (1.5, 0.3) -- (1.5,0.8);

\draw (1.5, 0.8) node[anchor=west] {$\lambda_i$};
\end{tikzpicture}
\end{subfigure}%
%
\caption{Agent period~1 value distribution under the exponential type model.
\label{fig:pdf_exp}} \vsq{-0.5em}
\end{figure}
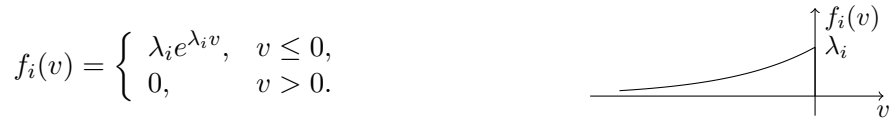

\subsection{Two-Period Mechanisms}

We consider {\em two-period mechanisms}, denoted as $\mech = (\reportSet, x, \tzero, \tone)$, and following the timeline suggested by \citet{ma2019contingent}. The mechanisms can, in general, involve both a base payment that an agent will pay irrespective of her utilization decision as well as a penalty for no show. The mechanisms are defined for a general message space $\reportSet$ for reports, and with allocation rule $x$, 
and with each agent $i$ facing a base payment $\tzero_i(\report)$ and a penalty $\tone_i(\report)$. 

The timeline for a two-period mechanism is as follows: \medskip

\noindent{\em Period~$0$:} \vspace{-0.5em}
\begin{enumerate}[$\bullet$]
	\setlength\itemsep{0em}
	\item Each agent $i \in N$ reports $\report_i \in \reportSet$ to the mechanism based on knowledge of $\thetahat_i$.
	\item The mechanism allocates the right to use the resources to a subset of agents, $A \subseteq N $, with $|A| \leq m$, thus $x_i(\report) = 1$ for all $i \in A$ and $x_i(\report)= 0$ for all $i \notin A$. 
	\item For each agent $i\in N$, the  mechanism determines a base payment $\tzero_i(\report)$ that the agent will pay for sure.  For each assigned agent $i\in A$, the mechanism determines an additional penalty $\tone_{i}(\report)$ that will be charged for a no show.
%
\end{enumerate}

\vspace{-0.3em}

\noindent {\em Period~$1$}: \vspace{-0.5em}
\begin{enumerate}[$\bullet$]
	\setlength\itemsep{0em}
	\item The mechanism collects base payment $\tzero_i(\report)$ from each agent.
	\item Each allocated agent $i \in A$ privately observes the realized immediate value $v_{i}\1$ of $V_i\1$, and decides whether to use the resource based on this value, the future value $v_i\2$, and the no-show penalty $\tone_i(\report)$.
	\item The mechanism collects the penalty $\tone_i(\report)$ from any agent $i \in A$ who is a no show. 
\end{enumerate}

\begin{example}[$(m+1)\th$ price auction] The standard $m+1\th$ price auction for assigning $m$ resources can be described as a two-period mechanism, where the report space is $\mathcal{R} =\mathbb R$. Ordering agents in decreasing order of their reports, s.t. $\report_1 \geq \report_2 \geq \dots \report_n$ (breaking ties randomly), the allocation rule is $x_i(\report) = 1$ for all $i \leq m$, $x_i(\report) = 0$ for $i > m$. Each allocated agent is charged  $\tzero_i(\report) = \report_{m+1}$, and all other payments are zero. The $m+1\th$ price auction does not make use of any penalties.
\end{example}

\begin{example}[Generalized contingent second price mechanism] The
\emph{generalized contingent second price} (GCSP) mechanism~\cite{ma2019contingent} %
  for assigning $m$ homogeneous resources collects a single bid from each agent, allocates the right to use resource to the $m$ highest bidders, and charges the $m+1\th$ highest bid, \emph{but only if an allocated agent fails to use the resource}. Formally, $\mathcal{R} =\mathbb R$. Ordering the agents s.t. $\report_1 \geq \report_2 \geq \dots \report_n$ (breaking ties randomly), we have $x_i(\report) = 1$ for $i \leq m$, $x_i(\report) = 0$ for $i > m$, $\tone_i(\report) = \max_{i' \notin A} \report_{i'}$, and all other payments are 0.
\end{example}

\medskip

We assume risk-neutral, expected-utility maximizing agents, but with quasi-hyperbolic discounting for future utilities. 
%
Each assigned agent $i$ faces a \emph{two part payment} $(z,y)$, where $z$ is the \emph{penalty} the agent pays in period~1 in the case of no-show, and $y$ is the \emph{base payment} the agent always pays in period~1. 
When period~$1$ arrives and the agent learns
the realized immediate value $v_i\1$, 
she discounts the future by $\beta_i$, and 
makes decisions as if that she will gain utility $v_i\1 - y + \beta_i v_i\2 $ from using the resource, and $-y -z$ from not using the resource. Based on this, the agent uses the resource if and only if
\begin{align}
  v_i\1 - y + \beta_i v_i\2  \geq -y - z \Leftrightarrow v_i\1 \geq - z
  - \beta_i v_i\2,
\end{align}
breaking ties in favor of using the resource. Let $\one{\cdot}$ be the indicator function, and define $u_i(z)$,  the expected utility of the agent when facing penalty $z$, as
\begin{align}
	u_i(z) \triangleq \E{ (V_i\1 + v_i\2) \one{V_i\1 \geq -z - \beta_i v_i\2}} - z \Pm{V_i\1< -z - \beta_i v_i\2}. \label{equ:exp_util_z} 
\end{align}

The \emph{actual} expected utility of an allocated agent facing a two-part payment $(z,y)$ is $u_i(z) - y$. Under a two-period mechanism $\mech$, given report profile $\report$, agent $i$'s expected utility is $x_i(\report) u_i(\tone_i(\report)) - \tzero_i(\report)$.

An agent believes that she will make decisions as if she has present-bias factor $\betahat_i$, and will decide to use the resource in period~1 if and only if
\begin{align}
  v_i\1 \geq - z - \betahat_i v_i\2.
\end{align}
Therefore, an  
agent's {\em subjective expected utility} given penalty $z$ is:
\begin{align}
	\uhat_i(z)  \triangleq \E{ (V_i\1 + v_i\2) \one{V_i\1 \geq -z - \betahat_i v_i\2}} - z \Pm{V_i\1< -z - \betahat_i v_i\2}.
	\label{equ:exp_util_hat}	
\end{align}
We call $\uhat_i(z)$ the \emph{subjective expected utility function}. For sophisticated agents who are able to perfectly predict their future decisions (i.e. $\betahat_i = \beta_i$),  $\uhat_i(z)$  and $u_i(z)$ coincide.


We assume that if allocated, agents' decisions in period~1 are influenced by their present bias, but are otherwise rational. The interesting question is to study an agent's incentives regarding reports in period~$0$, which are made based on subjective expected utility  $\uhat_i(z) - y$.
%
%
For any vector $g = (g_1, \dots, g_n)$ and any $i \in N$, we denote $g_{-i} \triangleq (g_1, \dots, g_{i-1}, g_{i+1}, \dots, g_n)$.

\begin{definition}[Dominant strategy equilibrium] 
A two-period mechanism has a {\em dominant strategy equilibrium} (DSE) if for each agent $i \in N$, for any type $\theta_i$ satisfying (A1)-(A3), there exists a report $\report^\ast_i \in \reportSet$ such that $\forall \report_{i} \in \reportSet, ~\forall \report_{-i} \in \reportSet^{n-1}$, 
\begin{align*}
	x_i(\report^\ast_i, ~ \report_{-i}) \uhat_i(\tone_i(\report^\ast_i, ~ \report_{-i})) - \tzero_i(\report^\ast_i, ~ \report_{-i}) \geq x_i(\report_i, ~ \report_{-i}) \uhat_i(\tone_i(\report_i, ~ \report_{-i})) - \tzero_i(\report_i, ~ \report_{-i}).
\end{align*}
\end{definition}



Let $\report^\ast(\theta) = (\report^\ast_1, \dots, \report^\ast_n)$ denote a report profile under a DSE given type profile $\theta$.
\begin{definition}[Voluntary participation] 
A two-period mechanism satisfies {\em voluntary participation} (VP) if for each agent $i \in N$, for any type $\theta_i$ satisfying (A1)-(A3), and any report profile $\report_{-i} \in \reportSet^{n-1}$, 
\begin{align*}
	x_i(\report^\ast_i, ~ \report_{-i}) \uhat_i(\tone_i(\report^\ast_i, ~ \report_{-i})) - \tzero_i(\report^\ast_i, ~ \report_{-i}) \geq 0.
\end{align*}
\end{definition}

Voluntary participation requires that each agent has non-negative subjective expected utility under her dominant strategy, given that she makes present-biased but otherwise rational decisions in period $1$ if allocated, regardless of the reports made by the rest of the agents.
%
%
Voluntary participation allows an agent to have negative utility at the end of period 1. 
%

The expected revenue of a two-period mechanism $\mech$ is the total
expected payment made by the agents 
in the DSE, assuming present-biased but otherwise rational decisions in period $1$:
\begin{align}
	\rev_\mech(\theta) \triangleq &
	\sum_{i \in N } \left( \tzero_i(\report^\ast) + 
	x_i(\report)\tone_{i}(\report^\ast) \Pm{ V_{i}\1 < - \tone_{i}(\report^\ast)  - \beta_i v_i\2}\right).  	
\end{align}

\begin{definition}[No deficit] A two-period mechanism satisfies {\em no deficit} (ND) if, for any type profile $\theta$ that satisfies (A1)-(A3), the expected revenue is non-negative: $\rev_\mech(\theta) \geq 0$.
\end{definition}

The \emph{utilization} achieved by mechanism $\mech$ 
is the expected number of resources used by the assigned agents in the DSE: 
%
%
\begin{align}
	\ut_\mech(\theta) \triangleq \sum_{i \in N} x_i(\report^\ast) \Pm{V_{i}\1 \geq - \tone_i (\report^\ast) - \beta_i v_i\2}. 
\end{align}


The expected \emph{social welfare} achieved by mechanism $\mech$ is the total expected value derived by agents from using the resources:
\begin{align}
	\sw_\mech(\theta) \triangleq \sum_{i \in N} x_i(\report^\ast) \E{ (V_i\1 + v_i\2) \one{V_i\1 \geq - \tone_i(\report^\ast)- \beta_i v_i\2} }.   \label{equ:social_welfare}
\end{align}

Our objective is to design mechanisms that maximize expected social
welfare. We do not consider monetary transfers as part of the social welfare function. The reason $\tone_i(\report^\ast)$ appears in \eqref{equ:social_welfare} is that it affects decisions of the allocated agents in period $1$.

\section{The Two-Bid Penalty Bidding Mechanism} \label{sec:two_bid_csp}

In this section, we introduce the two-bid penalty bidding mechanism, and prove that agents have simple dominant strategies, regardless of their value distributions, levels of present bias, or degrees of sophistication.

\begin{definition}[Two-bid penalty bidding mechanism] \label{def:two_bid_penalty_bidding}
The \emph{two-bid penalty bidding mechanism} (2BPB) collects bids $\bmax = (\bmax_i, \dots, \bmax_n)$ from agents in period~0, 
%
and reorders  agents in decreasing order of $\bmax_i$ such that $\bmax_1 \geq \bmax_2 \geq \dots \geq \bmax_n$ (breaking ties randomly).
\begin{enumerate}[$\bullet$]
	\setlength\itemsep{0em}
	\item Allocation rule: $x_{i}(b) = 1$ for $i \leq m$, $x_i(b) = 0$ for $i > m$.
	\item Payment rule: the mechanism announces $\bmax_{m+1}$, elicits a second bid $\bmin_i \geq \bmax_{m+1}$ from each assigned agent $i \leq m$, and sets $\tone_i(b) = \bmin_i$. $\tone_i(b) = 0$ for all $i > m$, and $\tzero_i(b) = 0$ for all $i \in N$.
\end{enumerate}
\end{definition}


The 2BPB mechanism first asks agents to bid on the maximum penalties they are willing to accept for the option to use the resource for free, and assigns the resources to the highest bidders. The mechanism then asks each assigned agent to bid a penalty that is weakly higher than the $m+1\th$ bid, representing the amount she would like 
to be charged in case of a no-show.\footnote{Instead of using two rounds of bidding, we may also consider a direct revelation mechanism, where agents report their private information $\thetahat_i$, 
with which the mechanism determines the assignment and the contingent payments.
%
}

To establish the dominant strategy equilibrium under the 2BPB mechanism, we first prove some useful properties of agents' subjective expected utility function $\uhat_i(z)$.

\begin{restatable}{lemma}{lemmaExpUtility} \label{lem:exp_u}  
Given an agent with any type $\theta_i$ that satisfies (A1)-(A3), 
the agent's subjective expected utility $\uhat_i(z)$ as a function of the penalty $z$ satisfies:
\begin{enumerate}[(i)]
	\setlength\itemsep{0em}
	\item $\uhat_i(0) \geq 0$, and $\lim_{z \rightarrow +\infty} \uhat_i(z) \leq \E{V_i}$.
	\item $\uhat_i(z)$ is right continuous and upper-semi-continuous.
	%
\end{enumerate}
\end{restatable}


\begin{proof} 

We first prove part (i). $\uhat_i(0) \geq 0$ holds given \eqref{equ:exp_util_hat} and the fact that $\betahat_i \leq 1$ and $w_i \geq 0$. For the limit as $z \rightarrow +\infty$, observe that $\uhat_i(z)$ can be rewritten as:
\begin{align}
	 \uhat_i(z) = & \E{ (V_i\1 + \betahat_i v_i\2) \one{(V_i\1+\betahat_i v_i\2) \geq -z }} +  \notag \\
	& (1 - \betahat_i) v_i\2 \Pm{(V_i\1+\betahat_i v_i\2) \geq -z}  - z \Pm{(V_i\1 + \betahat_i v_i\2)< -z }  \notag  \\
	=& \E{\max\{V_i\1 + \betahat_i v_i\2, ~-z\}} + (1 - \betahat_i) v_i\2 \Pm{(V_i\1+\betahat_i v_i\2) \geq -z}. \label{equ:uhat_rewritten}
\end{align}
By the monotone convergence theorem, the first term of \eqref{equ:uhat_rewritten} converges to $ \E{V_i\1 + \betahat_i v_i\2} = \E{V_i\1} + \betahat_i v_i\2$ as $z \rightarrow +\infty$. The second term is upper bounded by $(1 - \betahat_i) v_i\2$, therefore we get $\lim_{z \rightarrow +\infty} \uhat_i(z) \leq \E{V_i\1} + v_i\2 = \E{V_i}$. 

\smallskip

For part (ii), $\max\{V_i\1 + \betahat_i v_i\2, ~-z\}$ is a continuous function in $z$, therefore its expectation $\E{\max\{V_i\1 + \betahat_i v_i\2, ~-z\}}$ is also continuous in $z$. 
$\Pm{(V_i\1+\betahat_i v_i\2) \geq -z}$ is right continuous, implying the right continuity of $\uhat_i(z)$. The upper semi-continuity (i.e. $\lim_{z \uparrow z^\ast} \uhat_i(z) \leq \uhat_i(z^\ast)$ for all $z^\ast \geq 0$) holds because of the fact that $(1 - \betahat_i) v_i\2 \geq 0$, and that $\Pm{(V_i\1+\betahat_i v_i\2) \geq -z}$ is upper semi-continuous. 
%
\end{proof}

\citet{ma2019contingent} had earlier proved that for a rational agent without present bias, her expected utility as a function of the penalty 
is continuous, convex, and monotonically decreasing. These properties no longer hold for present-biased agents, since a higher penalty may incentivize an agent to use the resource more optimally, resulting in a higher expected utility.

\smallskip

For any penalty $z$, we define $\Uhat_i(z)$ as agent $i$'s highest subjective expected utility 
for the best choice of penalty, assuming this penalty 
must be at least $z$:
\begin{align}
	\Uhat_i(z) = \sup_{z' \geq z} \uhat_i(z'). \label{equ:max_U}
\end{align}


The following lemma 
proves the continuity and monotonicity of $\Uhat_i(z)$, together with the existence of a zero-crossing for $\Uhat_i(z)$.
This zero-crossing point is the maximum penalty an agent will accept,
in the case that this agent can choose to be charged any
penalty weakly larger than this penalty.

\begin{restatable}{lemma}{lemmaMaxExpUtility} \label{lem:max_exp_u}  
Given any agent with type $\theta_i$ that satisfies (A1)-(A3), 
the agent's subjective expected utility $\Uhat_i(z)$ as a function of the minimum penalty $z$ satisfies:
\begin{enumerate}[(i)]
	\setlength\itemsep{0em}
	\item $\Uhat_i(z)$ is continuous and monotonically decreasing in $z$.
	\item There exists a zero-crossing $\zc_i$ s.t. $\Uhat_i(\zc_i) = 0$ and $\Uhat_i(z) < 0$ for all $z > \zc_i$. 
\end{enumerate}
\end{restatable}

\begin{proof} 
For part (i), the monotonicity of $\Uhat_i(z)$ is obvious, and the continuity 
is implied by the right continuity of $\uhat_i(z)$ as shown in Lemma~\ref{lem:exp_u}. 
For part (ii), Lemma~\ref{lem:exp_u} and assumption (A3) imply $\lim_{z \rightarrow \infty} \uhat_i(z) \leq \E{V_i} < 0$. Therefore, there exists $Z \in \setR$ s.t. $\uhat_i(z) < 0$ for all $z \geq Z$. As a result, $\Uhat_i(z) < 0$ holds for all $z \geq Z$, 
and the monotonicity and continuity of $\Uhat_i(z)$ then imply that the following supreme exists: 
\begin{align*}
	\zc_i \triangleq \sup\{ z \in \setR ~|~ \Uhat_i(z) \geq 0\}, 
\end{align*}
and that we must have $\Uhat_i(\zc_i) = 0$ and $\Uhat_i(z) < 0$ for all $z > \zc_i$. 
\end{proof}


The following example illustrates the expected utility functions of an agent with $(c_i, p_i)$ type (see Example~\ref{ex:vipi}), and shows that there may not exist a DSE under the CSP mechanism.

\begin{example} \label{ex:vipi_expected_utility} Consider a sophisticated agent whose type follow the $(c_i, p_i)$ model, 
who is assigned a resource and charged no-show penalty $z$. With probability $1 - p_i$, the agent is not able to use the resource, and has to pay the penalty. With probability $p_i$, the agent is able to use the resource at a cost of $c_i$, but will use the resource if and only if $\beta_i w_i - c_i \geq -z \Leftrightarrow z \geq c_i - \beta_i w_i$. Therefore, the agent's expected utility as a function of the no-show penalty is of the form: 
\begin{align*}
	u_i(z) = \pwfun{-z, & \txtif  z < c_i - \beta_i w_i, \\ 
		(w_i - c_i)p_i - (1-p_i)z, & \txtif z \geq c_i - \beta_i w_i,}
\end{align*}
and $\uhat_i(z) = u_i(z)$ holds for all $z \geq 0$ since the agent is sophisticated. 
Figure~\ref{fig:vipi_util_u} illustrates $\uhat_i(z)$ for an agent with $c_i - \beta_i w_i > 0$. Intuitively, $c_i - \beta_i w_i$ is the minimum penalty the agent needs to be charged so that she will use the resource when she is able to. When $z < c_i - \beta_i w_i$, the agent ends up always paying the penalty, which is too small to incentivize utilization. $\Uhat_i(z)$ of this agent is as shown in Figure~\ref{fig:vipi_util_U}. The maximum penalty the agent is willing to accept is $\zc_i = (w_i - c_i)p_i/(1-p_i)$.

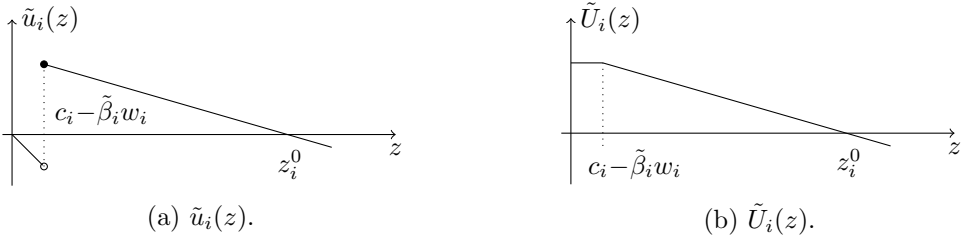
\begin{figure}[t!]
\centering
\begin{subfigure}[t]{0.45\textwidth}
	\centering
\begin{tikzpicture}[scale = 0.85][font = \small]

\draw[->] (-0.15,0) -- (6,0) node[anchor=north] {$z$};

\draw[->] (0,-0.8) -- (0, 1.8) node[anchor=west] {$\uhat_i(z)$};

\draw[-] (0, 0) -- (0.48,-0.48);
\draw[-] (0.5, 1.1) -- (5,-0.2);

\draw[dotted]	(0.5,-0.5) -- (0.5, 1.1);

\draw (0.5,-0.5) circle (1.5pt);
\filldraw [black] (0.5, 1.1) circle (1.5pt);

\draw	(1.4, 0) node[anchor=south] {$c_i \hspace{-0.2em} - \hspace{-0.2em} \betahat_i w_i$}
		(4.35,-0.1) node[anchor=north] {$\zc_i$};
\end{tikzpicture}

\caption{$\uhat_i(z)$. \label{fig:vipi_util_u}}
\end{subfigure}%
\begin{subfigure}[t]{0.45\textwidth}
	\centering
\begin{tikzpicture}[scale = 0.85][font=\small]

\draw[->] (-0.15,0) -- (6,0) node[anchor=north] {$z$};

\draw[->] (0,-0.8) -- (0, 1.8) node[anchor=west] {$\Uhat_i(z)$};

\draw[-] (0, 1.1) -- (0.5, 1.1) -- (5,-0.2);

\draw[dotted]	(0.5,-0.2) -- (0.5, 1.1);

\draw	(1,-0.1) node[anchor=north] {$c_i \hspace{-0.2em} - \hspace{-0.2em} \betahat_i w_i$}
		(4.35,-0.1) node[anchor=north] {$\zc_i$};
		

\end{tikzpicture}
\caption{$\Uhat_i(z)$. \label{fig:vipi_util_U}}
\end{subfigure}%
\caption{Subjective expected utility functions of a sophisticated agent with $(c_i, p_i)$ type, with $c_i - \betahat_i w_i > 0$.  \label{fig:vipi_example_utilities}} 
\end{figure}

There is no dominant strategy for this agent under the CSP mechanism. Consider the assignment of a single resource. If the highest bid among the rest of the agents satisfies $\max_{i' \neq i} b_{i'} \in [c_i - \betahat_i w_i, \zc_i)$, the  agent gets positive utility from bidding $b_i = \zc_i$, getting allocated and charged penalty $\max_{i' \neq i} b_{i'}$. However, if $\max_{i' \neq i} b_{i'} < c_i - \betahat_i w_i$, bidding $b_i = \zc_i$ results in negative utility--- the agent will be allocated, but charged a penalty that is too small to overcome her present bias. In this case, the agent is better off bidding $b_i = 0$ and get zero utility. 
\qed
\end{example}

We now state and prove the main theorem of this paper.

\begin{restatable}[Dominant strategy equilibrium of the two-bid
  penalty bidding mechanism]{theorem}{thmDSE} \label{thm:dse_two_bid_mech}  
Assuming (A1)-(A3), under the two-bid penalty bidding mechanism, it is a dominant strategy for each agent $i \in N$ to bid $\bmax_i^\ast = \zc_i$. If agent $i$ is assigned a resource and given a minimum penalty $\underline{z}$, it is then a dominant strategy to bid $\bmin_i^\ast = \arg\max_{z \geq \underline{z}} \uhat_i(z)$. Moreover, the mechanism satisfies voluntary participation and no deficit. 
\end{restatable}

\begin{proof}

We first consider an agent who is assigned a resource and asked by the mechanism to bid an amount that is at least $\underline{z}$. The right continuity of $\uhat_i(z)$ (see Lemma~\ref{lem:exp_u}) implies that the highest utility $\Uhat_i(\underline{z})$ when the agent can choose any penalty weakly higher than $\underline{z}$ is achieved at $\arg \max_{z \geq \underline{z}} \uhat_i(z)$. 
Since whichever amount an agent bids as $\bmin_i$ will be the penalty she is charged by the mechanism, it is a dominant strategy to bid $\bmin_i^\ast = \arg \max_{z \geq \underline{z}} \uhat_i(z)$. 

Given that an assigned agent will get expected utility $\Uhat_i(\underline{z})$ when she is asked to choose a penalty $\bmin_i $ that is weakly above $\underline{z}$, $\Uhat_i(z)$ is effectively her expected utility function in the first round of bidding. With the monotonicity of $\Uhat_i(z)$ and the fact that the minimum penalty is determined by the $m+1\th$ highest bid, it is standard that an agent bids in DSE the highest ``minimum penalty to choose from'' that she is willing to accept, which is $\zc_i$.  
%
\end{proof}

The following example shows that the 2BPB mechanism can achieve better social welfare and utilization than the $m+1\th$ price auction by assigning to a ``better'' agent and charging a proper penalty as the commitment device. 

\begin{example} \label{ex:two_bid_better_than_sp} Consider the assignment of one resource to two sophisticated agents with $(c_i, p_i)$ types, where:
\begin{enumerate}[$\bullet$]
	\item $c_1 = 10$, $p_1 = 0.8$, $w_1 = 16$, $\beta_1 = \betahat_1 = 0.5$,
	\item $c_2 = 6$, $p_2 = 0.5$, $w_2 = 10$, $\beta_2 = \betahat_2 = 0.8$.
\end{enumerate}
When $z < c_1 - \beta_1 w_1 = 2$, agent $1$ never uses the resource. On the other hand, $c_2 - \beta_2 w_2 < 0$ means that agent~$2$ uses the resource with probability $p_2$ while facing any non-negative penalty. $u_i(z) = \uhat_i(z)$ for $i = 1,2$ since both agents are sophisticated, and the subjective expected utility functions of the two agents are as shown in Figure~\ref{fig:exmp_two_bid_better_than_SP}. 

\begin{figure}[t!]
\centering
\begin{tikzpicture}[scale = 1][font = \small]

\draw[->] (-0.15,0) -- (6,0) node[anchor=north] {$z$};

\draw[->] (0,-0.8) -- (0, 1.8) node[anchor=west] {$\uhat_i(z)$};

\draw[-] (0, 0) -- (0.48,-0.48);
\draw[-] (0.5, 1.1) -- (5,-0.2);

\draw[dashed]	(0,0.5) -- (1.5, -0.25);

\draw[dotted]	(0.5,-0.5) -- (0.5, 1.1);

\draw (0.5,-0.5) circle (1.5pt);
\filldraw [black] (0.5, 1.1) circle (1.5pt);

\draw	(1, -0.1) node[anchor=north] {$\zc_2$}
		(4.35,-0.1) node[anchor=north] {$\zc_1$};
		
\draw[-] (4.4,1.5) -- (5,1.5) node[anchor=west] {$\uhat_1(z)$};
\draw[dashed] (4.4,1) -- (5,1) node[anchor=west] {$\uhat_2(z)$};
		
\end{tikzpicture}
\caption{Expected utility functions of two agents in Example~\ref{ex:two_bid_better_than_sp}.  \label{fig:exmp_two_bid_better_than_SP}} 
\end{figure}
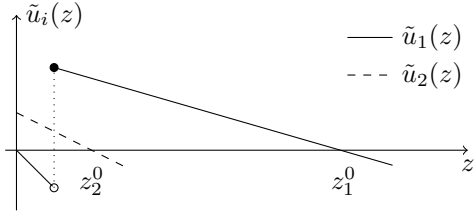

Under the second price auction, agents bid in DSE $b_{1, \txtSP}^\ast = \uhat_1(0) = 0$ and $b_{2, \txtSP}^\ast = \uhat_2(0) = (w_2 - c_2)p_2 = 2$, the values of the option to use the resource without any penalty (the free option to use the resource has no value to agent $1$ since she knows that she will never show up). Agent $2$ gets assigned the resource and charged no penalty, achieving social welfare $(w_2-c_2)p_2 = 2$ and utilization $p_2 = 0.5$.

Under the 2BPB mechanism, the agents bid in DSE $\bmax_1^\ast = \zc_1 = (w_1 - c_1)p_1/(1-p_1) = 24$, and $\bmax_2^\ast = \zc_2 = (w_2 - c_2)p_2/(1-p_2) = 4$. Agent~$1$ is therefore assigned and will bid $\bmin_1^\ast = 4$ when asked to choose a penalty weakly above $\bmax_2^\ast = 4$, since $\uhat_1(z)$ is monotonically decreasing in $z$ for $z \geq c_1 - \betahat_1 w_1 = 2$. The 2BPB mechanism achieves social welfare $(w_1-c_1)p_1 = 4.8$ and utilization $p_1 = 0.8$, both are higher than those under the second price auction.
\qed
\end{example}

\subsection{Discussion} \label{sec:discussion}

For fully rational agents with $\betahat_i = \beta_i = 1$, the subjective expected utility as a function of the penalty $\uhat_i(z)$ is monotonically decreasing in $z$, therefore $\uhat_i(z) = \Uhat_i(z)$ for all $z \in \setR$. In this case, the  equilibrium outcome under the 2BPB mechanism coincides with that under the $m+1\th$-price generalization of the CSP mechanism.

Since $\uhat_i(z)$ is what an agent considers while bidding, in period~0 a naive agent will bid as if she was rational with the same value distribution. In period~1, however, present bias will take effect, and the naive agent may make sub-optimal decisions. The actual expected utility a naive agent gets from participating in the $m+1\th$ price CSP or the 2BPB mechanisms, therefore, may be negative, despite the fact that she is willing to participate and believes she will get non-negative expected utility.

For two agents $i$ and $i'$ who are identical except that $\betahat_i > \betahat_{i'}$, we can prove that $\uhat_i(z) \geq \uhat_{i'}(z)$ holds for all $z \geq 0$. As a result, $\uhat_i(0) \geq \uhat_{i'}(0)$ and $\zc_i \geq \zc_{i'}$. This implies that an agent who believes that she is less present-biased  (i.e. with higher $\betahat_i$) will bid higher under both the 2BPB mechanism and the $m+1\th$ price auction. See Proposition~\ref{prop:bid_monotonicity} in Appendix~\ref{appx:bid_monotonicity} for more detailed discussions.

For rational agents without present bias, the CSP mechanism optimizes utilization among a large family of mechanisms with a set of desirable properties~\cite{ma2019contingent}. 
The 2BPB mechanism, however, does not provably optimize utilization for present-biased agents. The reason is that the actual present bias factor does not affect a naive agent's bid, and it is still possible for a very biased naive agent to be assigned
but rarely show up. On the other hand, the $m+1\th$ auction may not assign the resource to this agent, thus may achieve higher utilization and welfare (see Examples~\ref{exmp:two_bid_suboptimal_utilization_1} and~\ref{exmp:two_bid_suboptimal_utilization_2} in Appendix~\ref{appx:not_utilization_opt}).  


The 2BPB mechanism can also be generalized for assigning multiple heterogeneous resources $M = \{a, b, \dots, m\}$, 
where each agent $i \in N$ has a random value $V_{i,a} = V_{i,a}\1 + v_{i,a}\2$ for using each resource $a \in M$. $\uhat_{i,a}(z)$ and $\Uhat_{i,a}(z)$ can be defined similarly to \eqref{equ:exp_util_hat} and \eqref{equ:max_U}. The 2BPB mechanism can be generalized through the use of a minimum Walrasian equilibrium price mechanism, which computes the assignment and the minimum penalty each agent faces using $\{\Uhat_{i,a}(z)\}_{i \in N, a  \in M}$~\cite{demange1985strategy,DBLP:journals/ior/AlaeiJM16,ma2019contingent}. 
As a second step, each assigned agent is asked to report a weakly higher penalty that  she wants to be charged by the mechanism. 

\section{Simulation Results} \label{sec:simulations}

In this section, we adopt the exponential model (see Example~\ref{ex:exp_model}), and compare in simulation the social welfare and utilization achieved by different mechanisms and
benchmarks. 
Additional simulation results for the exponential model are presented in Appendix~\ref{appx:additional_simulations}, together with similar results when assuming the $(c_i, p_i)$ type model (see Example~\ref{ex:vipi}) or a uniform type model where agents' period~1 values are uniformly distributed.

For the exponential model, 
$\E{V_i} = -\lambda_i^{-1}+ \fixedV_i $, where $-\lambda_i^{-1}$ is the expected period~1 opportunity cost for using the resource.\footnote{The expected utility functions and dominant strategy bids under various type models are derived in Appendix~\ref{appx:derivations}.} 
We consider a type distribution in the population, where the value $w_i$ and the expected opportunity cost $\lambda_{i,a}^{-1}$ 
are uniformly distributed as $\lambda_{i}^{-1} \sim \mathrm{U}[0,~L]$ and $w_{i} \sim \mathrm{U}[0,~\lambda_{i}^{-1}]$. 
$w_{i} < \lambda_{i}^{-1}$ holds almost surely, thus assumptions (A1)-(A3) are satisfied. The results are not sensitive to the choices of $L$ in defining this type distribution, and we fix $L = 20$ for the rest of this section.

\subsection{Varying Resource Scarcity}  \label{sec:simulations_scarcity}

Fixing the number of resources at five, we study the impact of varying the scarcity of the resource, by varying the number of agents  from $2$ to $30$. 
We define the \emph{first best} as the highest achievable social welfare (or utilization) assuming full knowledge of agent types, and without violating voluntary participation or no deficit. 
The \emph{first-come-first-serve with fixed penalty mechanism} (FCFS) assumes a random order of arrival, with the effect of assigning to a random subset of at most $m$ agents who are willing to accept the penalty. 
We consider three levels of penalties for FCFS: 5, 2.5 and 0, where $5$ is equal to the expectation of the future value $w_i$.

\paragraph{Naive Agents} We first consider the scenario where all agents are naive. The present bias factor $\beta_i$ uniformly distributed on $[0,1]$, and all agents believe $\betahat_i = 1$. The average social welfare and utilization over 100,000 randomly generated profiles are as shown in Figure~\ref{fig:exp_naive_rand_beta}. 

\newcommand{\figScale}{0.85}

\begin{figure}[t!]
\centering
\begin{subfigure}[t]{0.45\textwidth}
	\centering
	\includegraphics[scale=\figScale]{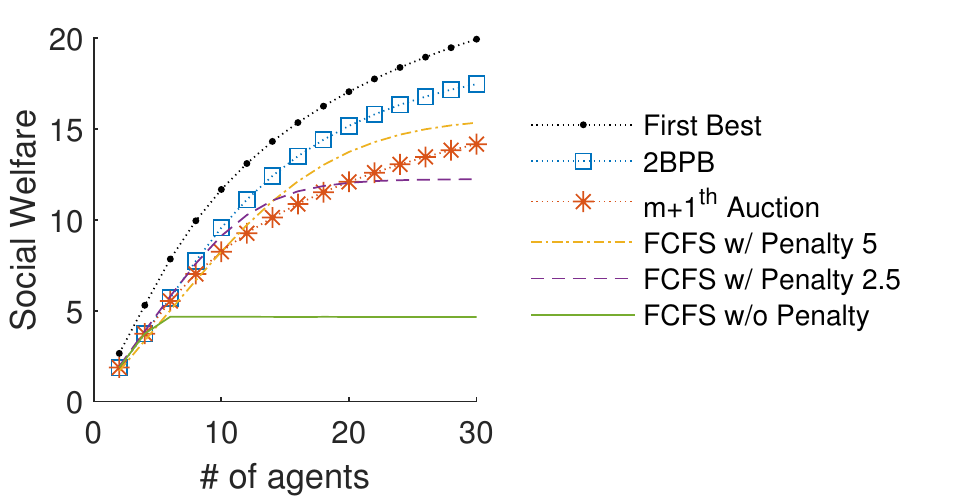}
	\caption{Social welfare. \label{fig:exp_naive_rand_beta_welfare}}
\end{subfigure}%
\hspace{2em}
\begin{subfigure}[t]{0.45\textwidth}
	\centering
 	\includegraphics[scale=\figScale]{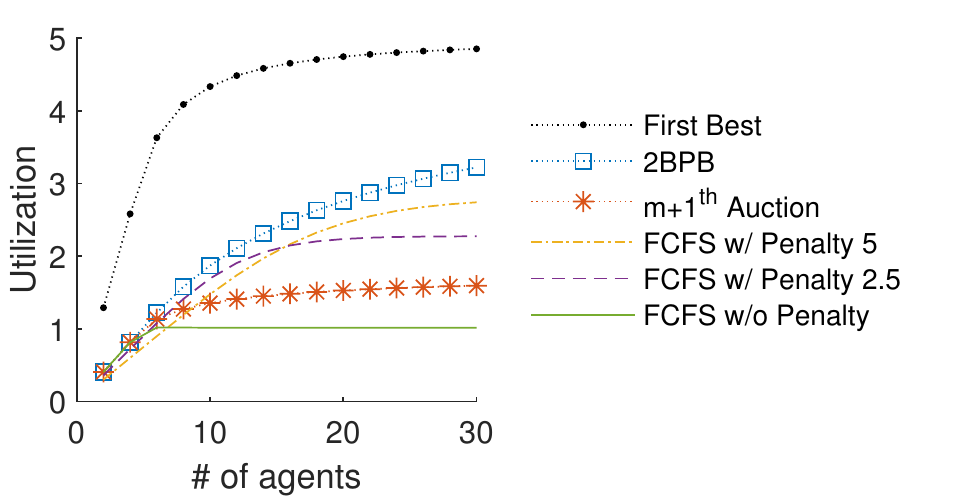}
	\caption{Utilization. \label{fig:exp_naive_rand_beta_utilization}}
\end{subfigure}%
\caption{Social welfare and utilization for naive agents with exponential types. 
\label{fig:exp_naive_rand_beta} 
}
\end{figure}

When the number of agents is small, the outcomes under 2BPB, the $m+1\th$ price auction, and the FCFS without penalty are similar, since all three effectively assign the resources to all agents, without charging any penalty. As the number of agents increases, the 2BPB mechanism achieves higher social welfare and substantially higher utilization than the $m+1\th$ price auction (which optimizes social welfare for rational agents without present bias), and does this without charging any payments from agents who do show up.

The 2BPB mechanism achieves higher welfare and utilization for economies of any size, and does not require any prior knowledge about the number of agents or their bias level or value distributions.
The FCFS mechanism (which are analogous to the reservation system widely used in practice), by comparison, requires careful adjustments of the fixed penalty level. A smaller penalty works fine when the number of agents is small but fails to keep up  as the economy becomes more competitive. A larger penalty outperforms the $m+1\th$ price auction for larger economies, but deters participation and leaves resources unallocated when the number of agents is small. 

\paragraph{Sophisticated Agents} Consider now fully sophisticated agents with $\betahat_i = \beta_i$, whose present-biased factors are distributed as $\beta_i \sim \mathrm{U}[0,1]$.
As the number of agents vary from $2$ to $30$, the average social welfare and utilization over 100,000 randomly generated economies are as shown in Figure~\ref{fig:exp_soph_rand_beta}. 

\begin{figure}[t!]
\centering
\begin{subfigure}[t]{0.45\textwidth}
	\centering
	\includegraphics[scale=\figScale]{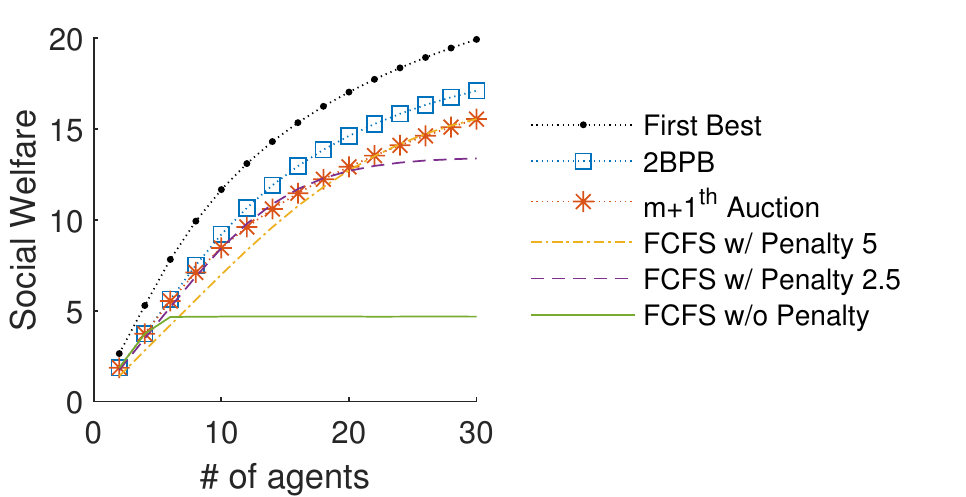}
	\caption{Social welfare. \label{fig:exp_soph_rand_beta_welfare}}
\end{subfigure}%
\hspace{2em}
\begin{subfigure}[t]{0.45\textwidth}
	\centering
 	\includegraphics[scale=\figScale]{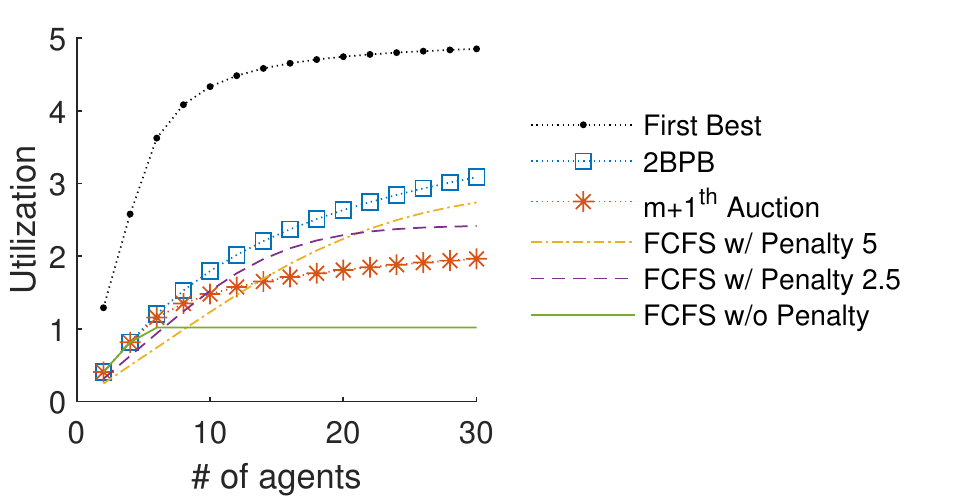}
	\caption{Utilization. \label{fig:exp_soph_rand_beta_utilization}}
\end{subfigure}%
\caption{Social welfare and utilization for sophisticated agents with exponential types. 
\label{fig:exp_soph_rand_beta} 
}
\end{figure}

As with the setting with naive agents, we can see that the 2BPB mechanism achieves higher welfare and utilization than the $m+1\th$ price auction, and that the performance of FCFS is very sensitive to the fixed penalty and the competitiveness of the economy. 
The $m+1\th$ price auction achieves higher welfare and utilization for sophisticated agents, in comparison to the setting with naive agents. This is because sophisticated agents are able to adjust their bids depending on their present bias level, and avoid the situation where a naive agent bids too much, gets assigned, but rarely show up, resulting in low utilization, welfare, and negative actual expected utility for the naive agent herself.

\smallskip

In Appendix\ref{appx:additional_simu_exp}, we present additional simulation results assuming all agents are fully rational ($\betahat_i = \beta_i = 1$) or partially naive (in which case we assume $\betahat_i \sim \mathrm{U}[\beta_i,~1]$).
The outcome for partially naive agents is between the outcome for fully naive agents and fully sophisticated agents. For 
rational agents, the 2BPB mechanism achieves slightly worse welfare than the $m+1\th$ price auction, which is provably optimal for this setting. The 2BPB mechanism, however, still achieves higher utilization and also a significantly better outcome than the FCFS benchmarks.

\subsection{Impact on Agents with Different Degrees of Bias} \label{sec:simulations_degree_of_bias}

In this section, we 
study the different outcomes for agents with different degrees of present bias. 
We assume the same type distribution 
as in the previous section, but fix the present-bias factor of each agent $i$ at $\beta_i = i/n$, where $n$ is the total number of agents--- the smaller an agent's index, the more present-biased an agent. 

\paragraph{Naive Agents} We first consider the scenario where all agents are naive. Fixing $n = 30$, for 1 million randomly generated economies, the average per economy welfare and usage (i.e. the probability of being assigned and showing up) of \emph{each agent} is as shown in Figure~\ref{fig:exp_naive_array_beta}. 
Under the first-best welfare and the first-best utilization, agents with different degrees of bias achieve the same welfare and utilization. This is because the agents all have the same distribution of $\lambda_i^{-1}$ and $w_i$, and only differ in their bias factor $\beta_i$. The full-information first best knows the exact types of agents, and adjusts the penalties accordingly, so that there is no difference between agents who are more or less biased. 
Note that naive agents behave in period~0 as if they were rational, therefore all agents bid in the same way despite their different degrees of bias, and therefore are assigned with the same probability.

\begin{figure}[t!]
\centering
\begin{subfigure}[t]{0.45\textwidth}
	\centering
	\includegraphics[scale=\figScale]{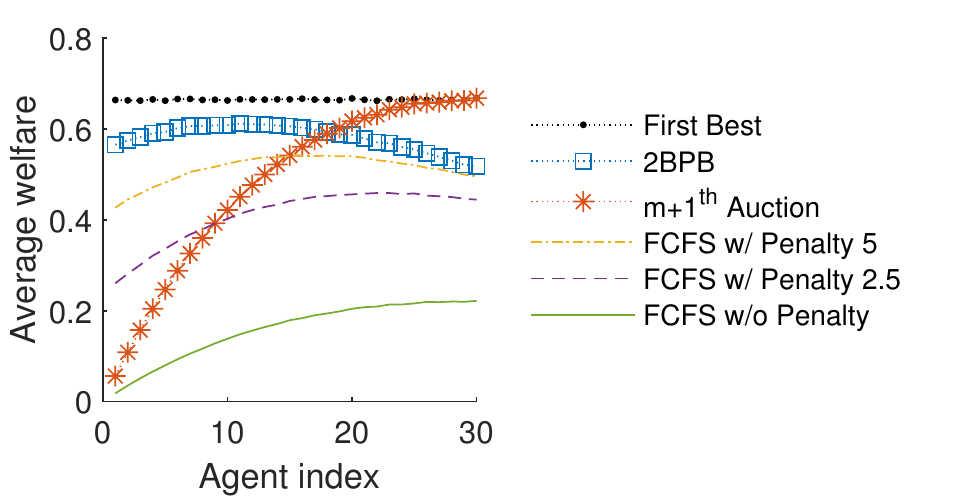}
	\caption{Average welfare. \label{fig:exp_naive_array_beta_welfare}}
\end{subfigure}%
\hspace{2em}
\begin{subfigure}[t]{0.45\textwidth}
	\centering
 	\includegraphics[scale=\figScale]{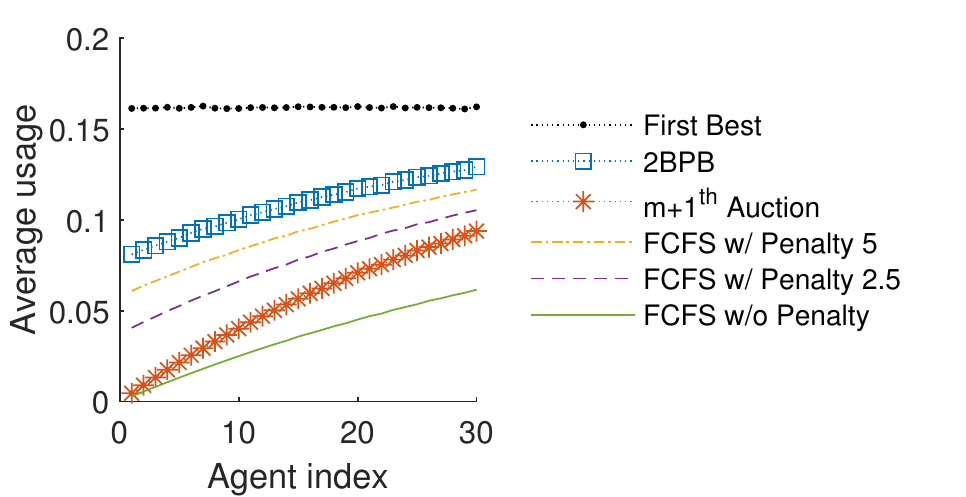}
	\caption{Average usage. \label{fig:exp_naive_array_beta_utilization}}
\end{subfigure}%
\caption{Average welfare and usage for naive agents with exponential types, fixing $\beta_i = i/n$. 
\label{fig:exp_naive_array_beta} 
}
\end{figure}

Figure~\ref{fig:exp_naive_array_beta_welfare} shows that the less biased agents (higher indices)  gain substantially higher welfare than the more biased agents (lower indices) under the $m+1\th$ price auction. By contrast, the 2BPB mechanism helps agents who are more biased to achieve substantially higher welfare than the outcome under the $m+1\th$ price auction, and at the same time slightly reducing the welfare for the least biased agents. This is because 
the least biased agents 
are able to make close to optimal decisions in period~$1$ by themselves, and charging a penalty sometimes leads to suboptimal utilization decisions.

From Figure~\ref{fig:exp_naive_array_beta_utilization}, we see that all agents have higher average usage under the 2BPB mechanism, 
and agents who are more biased achieve a higher gain 
in comparison with the $m+1\th$ price auction. 
Overall, the outcome under the 2BPB mechanism is subtantially more equitable for agents with all levels of bias. It is also worth noting that while naive agents do not see the value of commitment and  do not take any commitment device when offered~\cite{bryan2010commitment,beshears2011self}, the 2BPB mechanism is still able to help, 
since a commitment device is designed through the mechanism, and it is not an option to not accept a commitment.

\paragraph{Sophisticated Agents} For fully sophisticated agents with varying degrees of present bias, the average welfare and usage per economy of each agent are as shown in Figure~\ref{fig:exp_soph_array_beta}.

\begin{figure}[t!]
\centering
\begin{subfigure}[t]{0.45\textwidth}
	\centering
	\includegraphics[scale=\figScale]{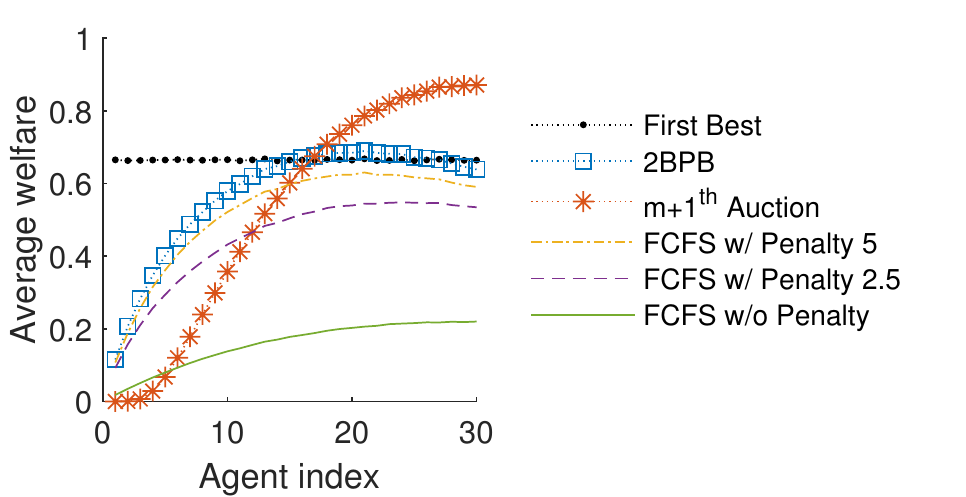}
	\caption{Average welfare. \label{fig:exp_soph_array_beta_welfare}}
\end{subfigure}%
\hspace{2em}
\begin{subfigure}[t]{0.45\textwidth}
	\centering
 	\includegraphics[scale=\figScale]{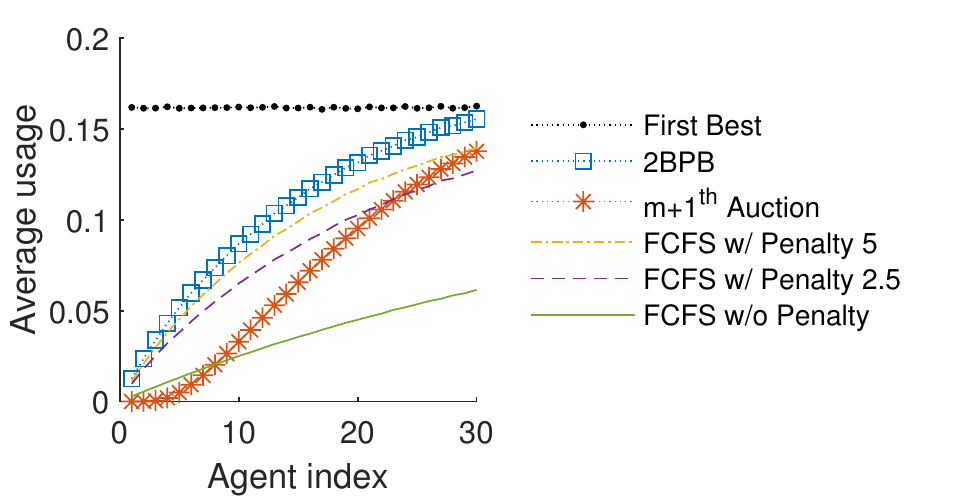}
	\caption{Average usage. \label{fig:exp_soph_array_beta_utilization}}
\end{subfigure}%
\caption{Average welfare and usage for sophisticates with exponential types, fixing $\beta_i=i/n$. 
\label{fig:exp_soph_array_beta} 
}
\end{figure}

The first observation is that under the $m+1\th$ price auction, the welfare and usage for the most biased agents are effectively zero, while the least biased agents achieve better welfare and utilization than the first-best outcome. This is because when the bids of sophisticated agents factor in the level of present bias, the more biased agents bid lower than the less biased agents (see Claim~\ref{prop:bid_monotonicity} in Appendix~\ref{appx:bid_monotonicity}), and therefore get assigned with lower probability.
The more biased sophisticated agents also bid lower under the 2BPB mechanism, and as a result the 2BPB mechanism is not able to achieve the same level of welfare for all agents. Nevertheless, it achieves large improvements for the more biased population compared to the $m+1\th$ price auction, and also higher welfare and better equity than the FCFS benchmarks. 
%
%

\subsection{Impact on Agents with Different Degrees of Sophistication}
\label{sec:simulations_degree_of_naivete}

In this section, we consider a population of agents with the same present-bias factor, but varying degree of sophistication. We assume the same distribution of the expected opportunity cost $\lambda_i^{-1}$ and the future value $w_i$ as in the earlier settings, but fix $\beta_i = 0.5$ and $\betahat_i = 1 - 0.5 i/n$ for all agents. The smaller an agent's index, the more naive she is about her present bias: agent $1$ has $\betahat_1$ close to $1$ thus is almost fully naive, whereas agent $n$ has $\betahat_n = 0.5 = \beta_n$ thus is fully sophisticated.

\begin{figure}[t!]
\centering
\begin{subfigure}[t]{0.45\textwidth}
	\centering
	\includegraphics[scale=\figScale]{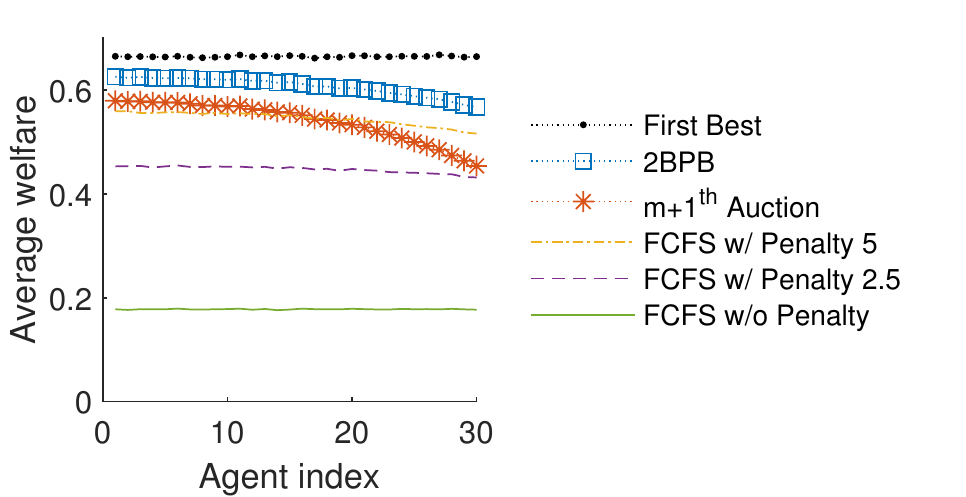}
	\caption{Average welfare. \label{fig:exp_fix_beta_array_betahat_welfare}}
\end{subfigure}%
\hspace{2em}
\begin{subfigure}[t]{0.45\textwidth}
	\centering
 	\includegraphics[scale=\figScale]{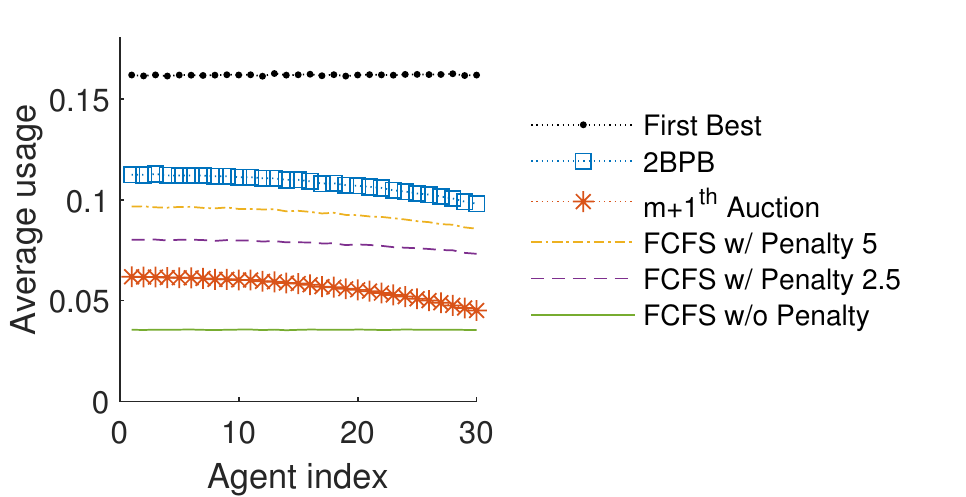}
	\caption{Average usage. \label{fig:exp_fix_beta_array_betahat_utilization}}
\end{subfigure}%
\caption{Welfare and usage for agents with exponential types and varying degrees of naivete. 
\label{fig:exp_fix_beta_array_betahat}  \vspace{-0.5em}
}
\end{figure}

Fixing $n = 30$, the per-economy welfare and usage for each agent averaged over 1 million randomly generated economies are as shown in Figure~\ref{fig:exp_fix_beta_array_betahat}. All agents achieve higher welfare and usage under the 2BPB mechanism, in comparison to the $m+1\th$ price auction and the FCFS mechanisms. The outcome under the 2BPB mechanism is again more equitable than that under the $m+1\th$ price auction. It is curious that agents who are less naive (higher indices) have lower welfare and usage. This is because the more sophisticated agents can better predict their future suboptimal decisions, and as a result bid lower and accept fixed penalties with lower probability. It is indeed the case that the more sophisticated agents achieve slightly higher expected utility.

\section{Conclusion} \label{sec:conclusion}

We propose the two-bid penalty-bidding mechanism for resource allocation in the presence of uncertain future values and present bias.  We prove the existence of a simple dominant strategy equilibrium, regardless of an agent's value distribution, level of  present bias, or degree of sophistication. Simulation results show  that the mechanism improves utilization and achieves higher welfare and better equity in comparison with mechanisms broadly used in practice as well as mechanisms that are welfare-optimal for settings without present bias.

In future work, it will be interesting to conduct empirical studies to better understand people's behavior in settings such as exercise studios and events, with the goal of separating the effect on utilization of uncertainty from that of present bias.  Another interesting direction is to generalize the model to allow for more than two time periods, where agents may arrive asynchronously, when uncertainty unfolds over time, and where resources can be re-allocated.


\bibliography{CSP_refs}

\newpage

\appendix

\noindent{}\textbf{\huge{Appendix}}

\bigskip

\noindent{}Appendix~\ref{appx:additional_simulations} provides additional simulation results. Appendix~\ref{appx:additional_discussion} provides additional discussions and examples omitted from the body of the paper. Appendix~\ref{appx:derivations} derives for various type models the expected utility function and the dominant strategy equilibrium under different mechanisms.

\section{Additional Simulation Results} \label{appx:additional_simulations}

This section presents additional simulation results for the exponential type model that are omitted from the body of the paper, as well as results for the $(c_i, p_i)$ type model and a uniform type model, which we introduce in Appendix~\ref{appx:additional_simu_uniform}.

\subsection{Additional Results for Exponential Model} \label{appx:additional_simu_exp}

We first consider the same setup as analyzed in Section~\ref{sec:simulations}, where agents have exponential types, and there are $m = 5$ homogeneous resources to assign. We present the results as the number of agents varies, for settings where agents are all 
fully rational, or where agents are partially naive. 

\paragraph{Fully Rational Agents} 
Figure~\ref{fig:exp_rational_rand_beta} shows the average welfare and utilization of 100,000 randomly generated economies, assuming all agents are fully rational with $\betahat_i = \beta_i = 1$. The 2BPB mechanism achieves slightly worse welfare than the $m+1\th$ price auction, which is provably optimal for this setting. The 2BPB mechanism, however, still achieves higher utilization, and also robustly outperforms the FCFS mechanisms.

\begin{figure}[hpbt!]
\centering
\begin{subfigure}[t]{0.45\textwidth}
	\centering
	\includegraphics[scale=\figScale]{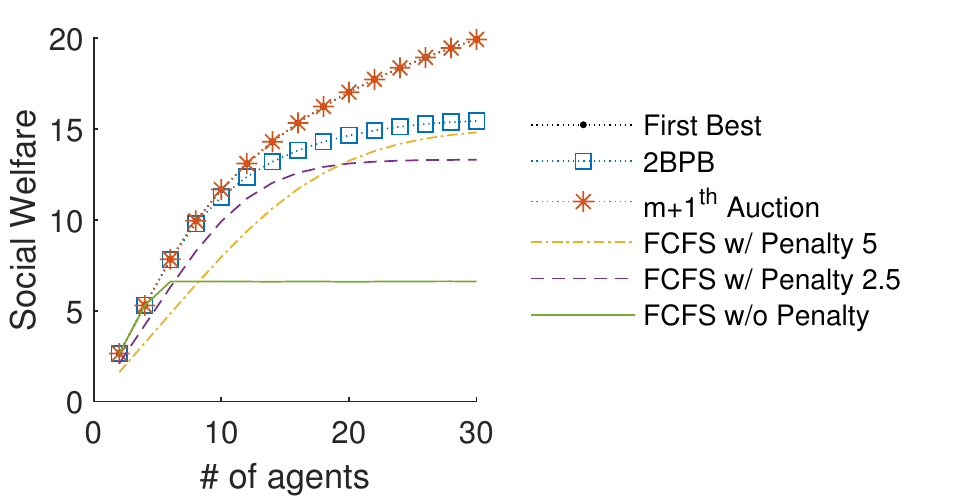}
	\caption{Social welfare. \label{fig:exp_rational_rand_beta_welfare}}
\end{subfigure}%
\hspace{2.5em}
\begin{subfigure}[t]{0.45\textwidth}
	\centering
 	\includegraphics[scale=\figScale]{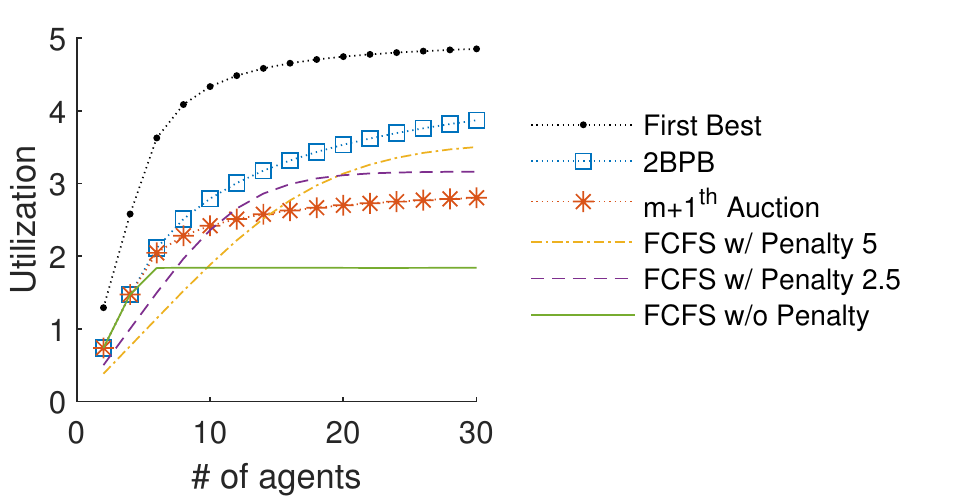}
	\caption{Utilization. \label{fig:exp_rational_rand_beta_utilization}}
\end{subfigure}%
\caption{Social welfare and utilization for rational agents with exponential types. 
\label{fig:exp_rational_rand_beta} 
}
\end{figure}

\paragraph{Partially Naive Agents}  We now consider partially naive agents, where each agent has bias factors independently drawn according to $\beta_i \sim \mathrm{U}[0,~ 1]$, and $\betahat_i \sim \mathrm{U}[\beta_i,~1]$. The average welfare and utilization over 100,000 random economies are as shown in Figure~\ref{fig:exp_part_rand_beta}. The outcome is in between the fully sophisticated and the fully naive settings discussed in the body of the paper.

\begin{figure}[hpbt!]
\centering
\begin{subfigure}[t]{0.45\textwidth}
	\centering
	\includegraphics[scale=\figScale]{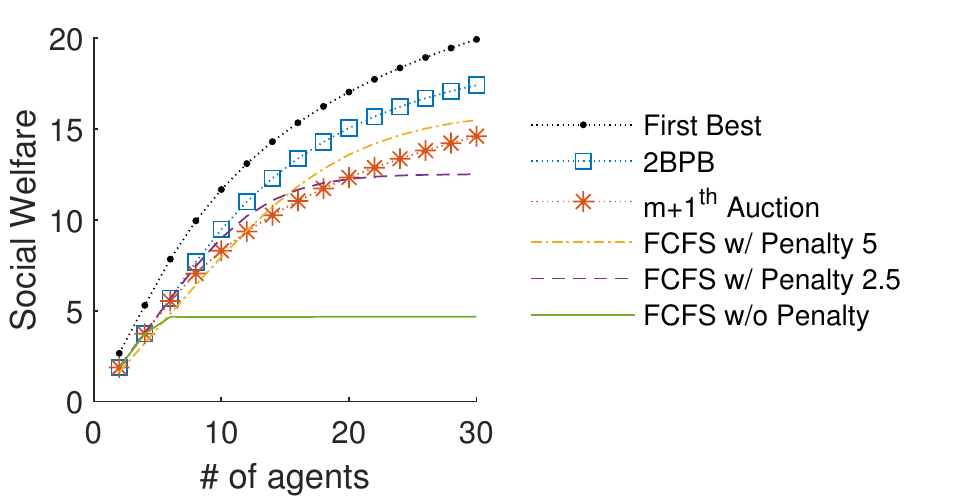}
	\caption{Social welfare. \label{fig:exp_part_rand_beta_welfare}}
\end{subfigure}%
\hspace{2.5em}
\begin{subfigure}[t]{0.45\textwidth}
	\centering
 	\includegraphics[scale=\figScale]{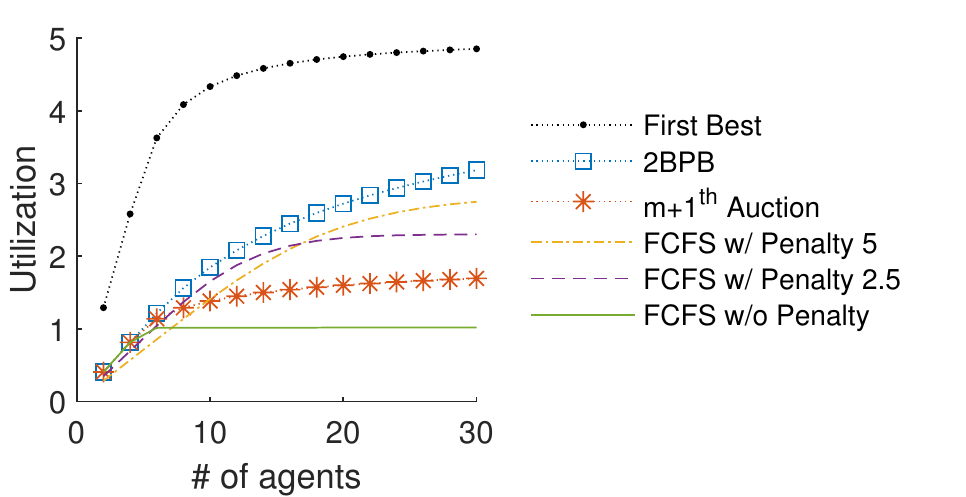}
	\caption{Utilization. \label{fig:exp_part_rand_beta_utilization}}
\end{subfigure}%
\caption{Social welfare and utilization for partially naive agents with exponential types. 
\label{fig:exp_part_rand_beta} 
}
\end{figure}


\subsection{The $(c_i, p_i)$ Type Model}

n this section, we compare via simulation different mechanisms and benchmarks for the $(c_i, p_i)$ type model (see Example~\ref{ex:vipi}).\footnote{Agents' expected utility functions and DSE bids are derived in Appendix~\ref{appx:derivations}.} We consider a type distribution  where the future value $w_i$, cost $c_i$, and probability of being able to show up $p_i$ are uniformly distributed:
\begin{align*}
	w_i &\sim \mathrm{U}[0,~L],\\
	c_i & \sim \mathrm{U}[0,~w_i], \\
	p_{i} &\sim \mathrm{U}[0,~1].
\end{align*}
With $c_{i} < w_i$  and $p_i \in (0,1)$ with probability~1, assumptions (A1)-(A3) are satisfied almost surely. The results are not sensitive to the choices of  parameter $L$, and we fix $L = 10$ for all results presented in the rest of this section. In this case, the expected value of $w_i$ is $5$.

\subsubsection{Varying Resource Scarcity}

Fixing the number of resources at $m = 5$, we first examine the outcomes under different mechanisms and benchmarks as the number of agents varies from $2$ to $30$. When all agents are naive with $\beta_i \sim \mathrm{U}[0,~1]$ and $\betahat_i = 1$. The average social welfare and utilization over 100,000 randomly generated profiles are as shown in Figure~\ref{fig:cipi_naive_rand_beta}. For economies with fully sophisticated agents, where present-biased factors are distributed as $\beta_i \sim \mathrm{U}[0,1]$, but $\betahat_i = \beta_i$ for all $i \in N$, the average social welfare and utilization are as shown in Figure~\ref{fig:cipi_soph_rand_beta}. We see trends similar to the results under the exponential type model presented in Section~\ref{sec:simulations_scarcity}.

\begin{figure}[t!]
\centering
\begin{subfigure}[t]{0.45\textwidth}
	\centering
	\includegraphics[scale=\figScale]{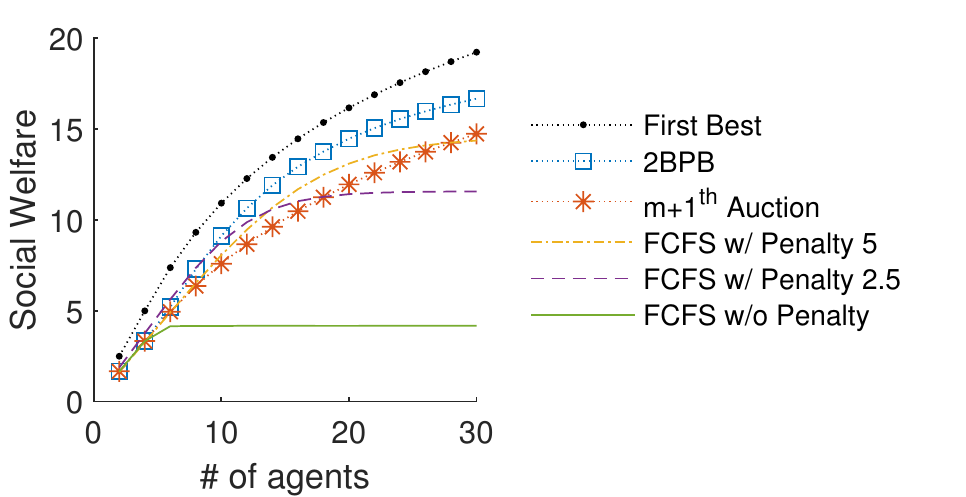}
	\caption{Social welfare. \label{fig:cipi_naive_rand_beta_welfare}}
\end{subfigure}%
\hspace{2em}
\begin{subfigure}[t]{0.45\textwidth}
	\centering
 	\includegraphics[scale=\figScale]{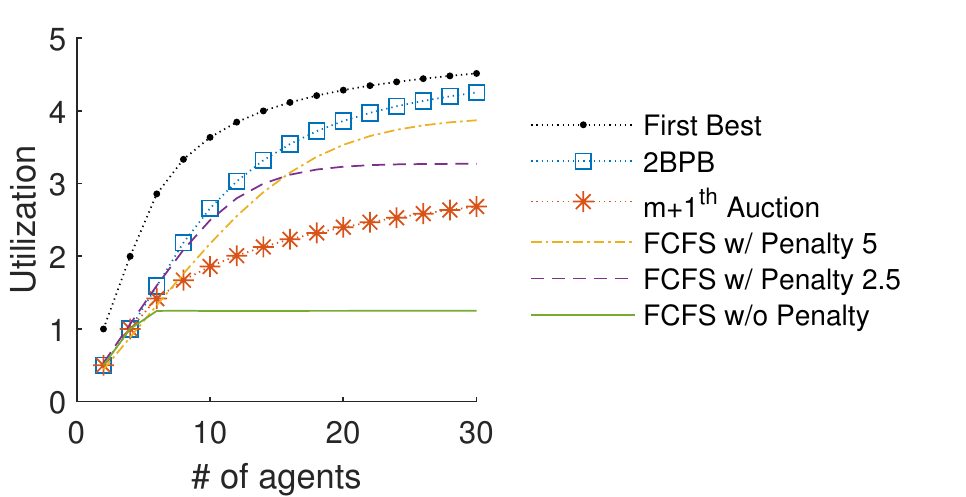}
	\caption{Utilization. \label{fig:cipi_naive_rand_beta_utilization}}
\end{subfigure}%
\caption{Social welfare and utilization for naive agents with $(c_i,p_i)$ types. 
\label{fig:cipi_naive_rand_beta} 
}
\end{figure}

\begin{figure}[t!]
\centering
\begin{subfigure}[t]{0.45\textwidth}
	\centering
	\includegraphics[scale=\figScale]{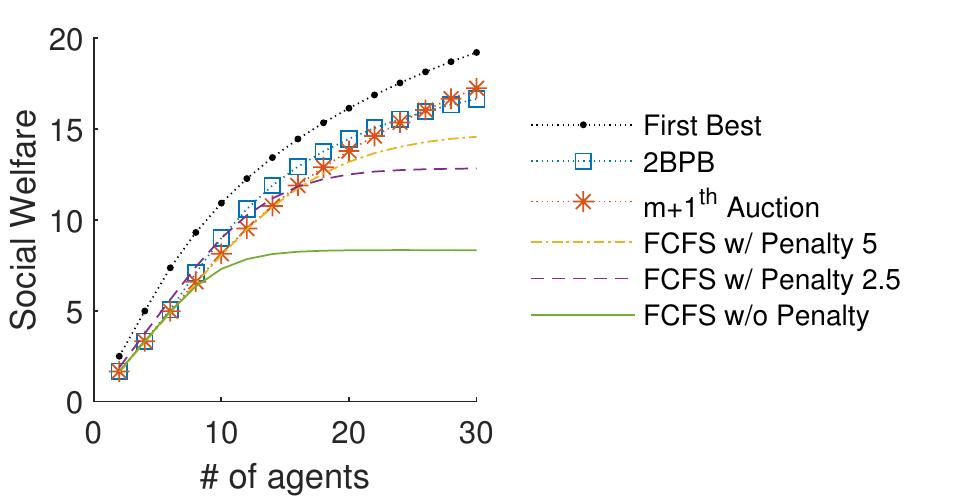}
	\caption{Social welfare. \label{fig:cipi_soph_rand_beta_welfare}}
\end{subfigure}%
\hspace{2em}
\begin{subfigure}[t]{0.45\textwidth}
	\centering
 	\includegraphics[scale=\figScale]{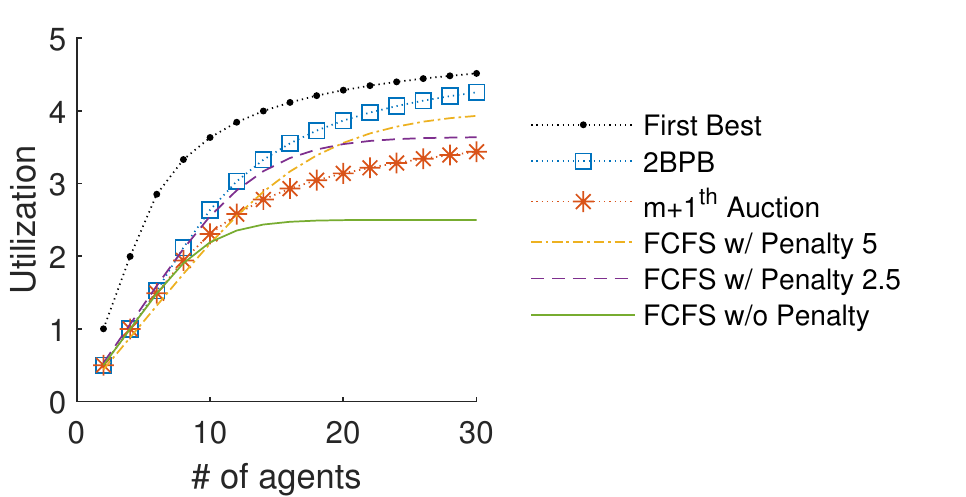}
	\caption{Utilization. \label{fig:cipi_soph_rand_beta_utilization}}
\end{subfigure}%
\caption{Social welfare and utilization for sophisticated agents with $(c_i,p_i)$ types. 
\label{fig:cipi_soph_rand_beta} 
}
\end{figure}

\subsubsection{Impact on Agents with Different Degrees of Bias}

We now consider a population of agents with the same distribution of $c_i$, $p_i$ and $w_i$ as in the previous setting, but where the total number of agents is fixed at $n=30$, and the present bias factor of  each agent $i$ is fixed at $\beta_i = i/n$. 
Assuming all agents are naive with $\beta_i = 1$, the average welfare and average usage of each agent (over 1 million randomly generated economies) is as shown in Figure~\ref{fig:cipi_naive_array_beta}. For the setting where all agents are fully sophisticated with $\betahat_i = \beta_i$, the average welfare and usage of each agent is as shown in Figure~\ref{fig:cipi_soph_array_beta}.

\begin{figure}[t!]
\centering
\begin{subfigure}[t]{0.45\textwidth}
	\centering
	\includegraphics[scale=\figScale]{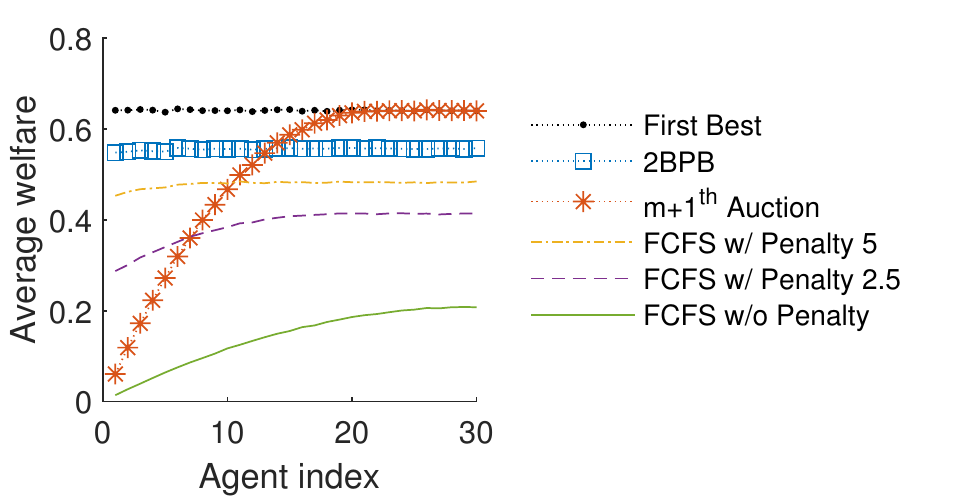}
	\caption{Average welfare. \label{fig:cipi_naive_array_beta_welfare}}
\end{subfigure}%
\hspace{2em}
\begin{subfigure}[t]{0.45\textwidth}
	\centering
 	\includegraphics[scale=\figScale]{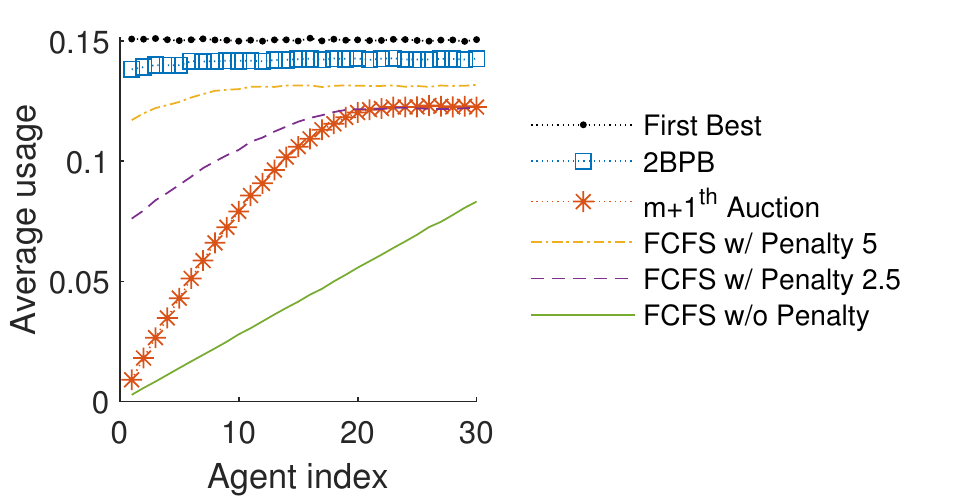}
	\caption{Average usage. \label{fig:cipi_naive_array_beta_utilization}}
\end{subfigure}%
\caption{Average welfare and usage for naive agents with  $(c_i,p_i)$  types, fixing $\beta_i = i/n$. 
\label{fig:cipi_naive_array_beta} 
}
\end{figure}

\begin{figure}[t!]
\centering
\begin{subfigure}[t]{0.45\textwidth}
	\centering
	\includegraphics[scale=\figScale]{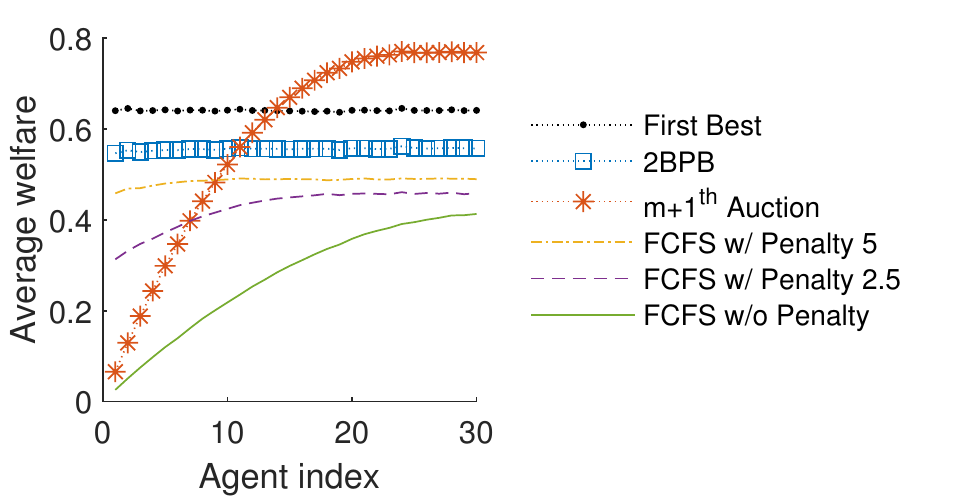}
	\caption{Average welfare. \label{fig:cipi_soph_array_beta_welfare}}
\end{subfigure}%
\hspace{1em}
\begin{subfigure}[t]{0.45\textwidth}
	\centering
 	\includegraphics[scale=\figScale]{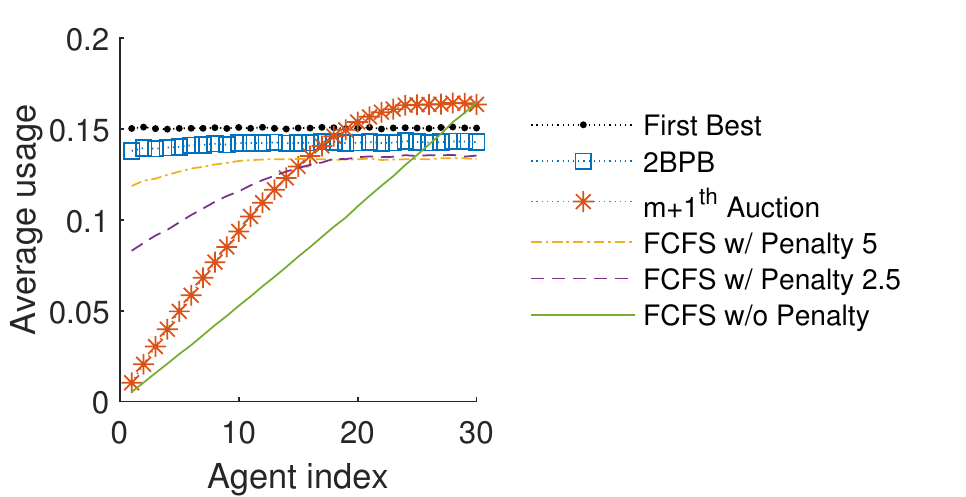}
	\caption{Average usage. \label{fig:cipi_soph_array_beta_utilization}}
\end{subfigure}%
\caption{Average welfare and usage for sophisticated agents with  $(c_i,p_i)$ types, fixing $\beta_i=i/n$. 
\label{fig:cipi_soph_array_beta} 
}
\end{figure}

We can see from Figures~\ref{fig:cipi_naive_array_beta} and~\ref{fig:cipi_soph_array_beta} that for both the naive and the sophisticated settings, the less biased agents (higher indices) achieve significantly higher average welfare and usage under the $m+1\th$ price auction and the FCFS mechanisms. In contrast, this disparity under the 2BPB mechanism is much smaller, and the most biased agents benefit the most under the 2BPB mechanism in comparison to the FCFS mechanisms. 


\subsection{Uniform Type Mode}  \label{appx:additional_simu_uniform}

In this section, we study the performance of various mechanisms and benchmarks for agents whose period~$1$ values are uniformly distributed as in the following Example~\ref{ex:uniform_model}. 

\begin{example}[Uniform model] \label{ex:uniform_model} 
In period~$1$, each agent incurs a uniformly distributed opportunity cost for using the resource, i.e. $V_i\1 \sim \mathrm{U}[-\alpha_i, 0]$. If the agent used a resource, she gains a expected future utility of $v_i\2 = \fixedV_i > 0$.
See Figure~\ref{fig:pdf_uniform}. $\E{V_i\1} = -\alpha_i/2 $, thus $\E{V_i} = -\alpha_i/2 + \fixedV_i$ and (A1)-(A3) are satisfied as long as $w_i < 1/\alpha_i/2$.
\end{example}

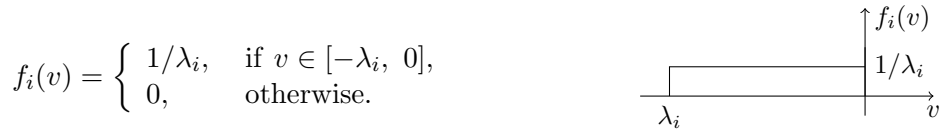
\begin{figure}[t!]
\centering
\begin{subfigure}[t]{0.45\textwidth}
	\centering
\small{
\begin{tikzpicture}[scale = 0.9][font=\normalsize]
	\draw (0,0) node  {$ f_i(v) = \pwfun{ 1/\lambda_i, & \txtif v \in [-\lambda_i,~0], \\
	0, & \text{ otherwise.}}$};		
	\draw (0,-0.6) node{{\color{white} some text}};
\end{tikzpicture}	
}
\end{subfigure}%
%
%
%
\begin{subfigure}[t]{0.45\textwidth}
	\centering
\begin{tikzpicture}[scale = 1.3][font=\small]

\draw[->] (-0.8,0.3) -- (2.2,0.3)  node[anchor=north] {$v$};
  
\draw[->] (1.5,0.1) -- (1.5,1.2);
\draw (1.5,1.1) node[anchor=west] {$f_i(v)$};

\draw[-] (-0.5, 0.3) -- (-0.5, 0.6) -- (1.5, 0.6);

\draw[-] (1.5, 0.3) -- (1.5,0.8);

\draw (-0.5, 0.3) node[anchor=north] {$\lambda_i$};
\draw (1.5, 0.6) node[anchor=west] {$1/\lambda_i$};
\end{tikzpicture}
\end{subfigure}%
%
\caption{Agent period~1 value distribution under the uniform type model. \label{fig:pdf_uniform}} \vsq{-0.5em}
\end{figure}

Agents' expected utility functions and DSE bids under the uniform model are derived in Appendix~\ref{appx:derivations}. We consider the following type distribution in the population, where $\alpha_i$ and $w_i$ are both uniformly distributed:
\begin{align*}
	\alpha_i &\sim \mathrm{U}[0,~L],\\
	w_i & \sim \mathrm{U}[0,~\alpha_i/2].
\end{align*}
With $w_i \in (0, \alpha_i/2)$ with probability~1, assumptions (A1)-(A3) are satisfied almost surely. The results are not sensitive to the choices of parameter $L$, and we fix $L = 20$ for all results presented in the rest of this section, in which case the average value of $w_i$ is $5$.

\subsubsection{Varying Resource Scarcity}

Fixing the number of resources at $m = 5$, we first examine the outcomes under various mechanisms and benchmarks as the number of agents varies from $2$ to $30$. For the scenario where all agents are naive, with $\beta_i \sim \mathrm{U}[0,~1]$ and $\betahat_i = 1$, the average social welfare and utilization over 100,000 randomly generated profiles are as shown in Figure~\ref{fig:uniform_naive_rand_beta}. For economies with fully sophisticated agents, where present-biased factors are distributed as $\beta_i \sim \mathrm{U}[0,1]$ and $\betahat_i = \beta_i$, the average social welfare and utilization are as shown in Figure~\ref{fig:uniform_soph_rand_beta}.

\begin{figure}[t!]
\centering
\begin{subfigure}[t]{0.45\textwidth}
	\centering
	\includegraphics[scale=\figScale]{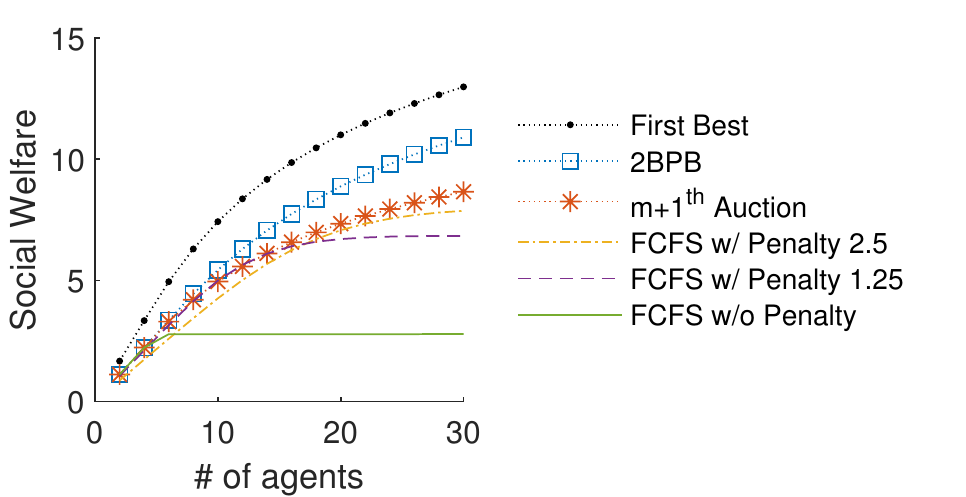}
	\caption{Social welfare. \label{fig:uniform_naive_rand_beta_welfare}}
\end{subfigure}%
\hspace{2em}
\begin{subfigure}[t]{0.45\textwidth}
	\centering
 	\includegraphics[scale=\figScale]{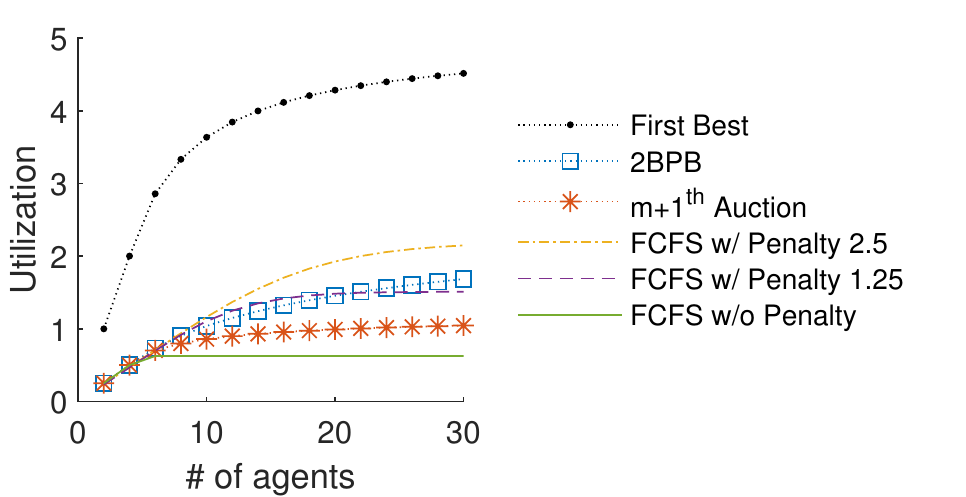}
	\caption{Utilization. \label{fig:uniform_naive_rand_beta_utilization}}
\end{subfigure}%
\caption{Social welfare and utilization for naive agents with uniform types. 
\label{fig:uniform_naive_rand_beta}
}
\end{figure}

\begin{figure}[t!]
\centering
\begin{subfigure}[t]{0.45\textwidth}
	\centering
	\includegraphics[scale=\figScale]{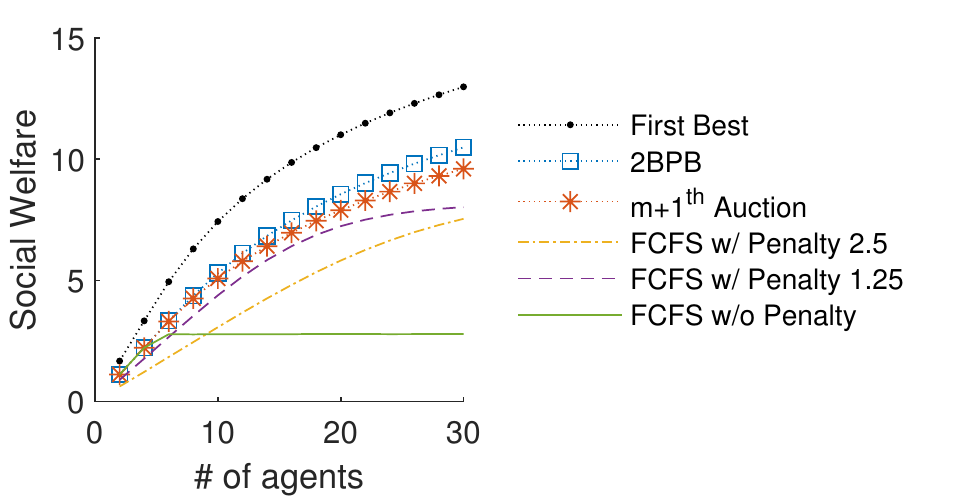}
	\caption{Social welfare. \label{fig:uniform_soph_rand_beta_welfare}}
\end{subfigure}%
\hspace{2em}
\begin{subfigure}[t]{0.45\textwidth}
	\centering
 	\includegraphics[scale=\figScale]{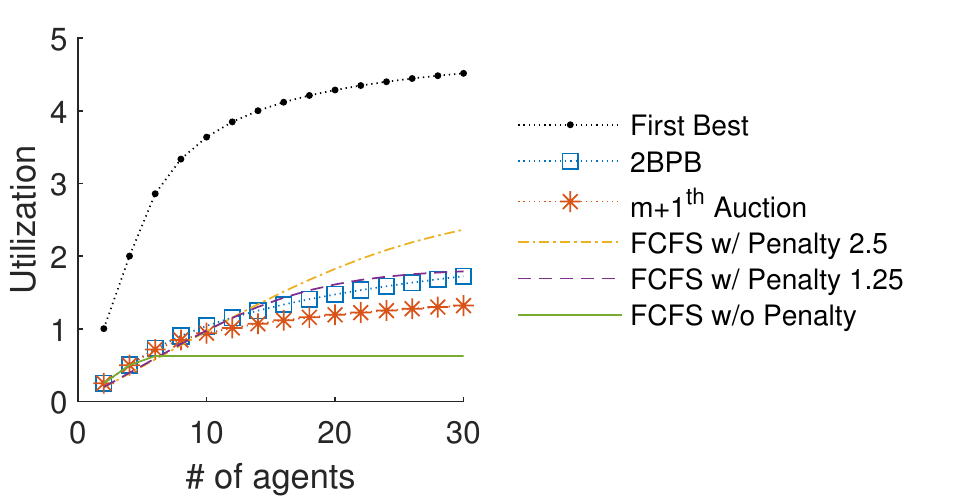}
	\caption{Utilization. \label{fig:uniform_soph_rand_beta_utilization}}
\end{subfigure}%
\caption{Social welfare and utilization for sophisticated agents with uniform types. 
\label{fig:uniform_soph_rand_beta} 
}
\end{figure}

\subsubsection{Impact on Agents with Different Degrees of Bias}

We now consider a population of agents with the same distribution of $\alpha_i$ and $w_i$ as in the previous setting, but where the total number of agents is fixed at $n=30$, and the present bias factor of agent each agent $i$ is fixed at $\beta_i = i/n$. Assuming all agents are naive, the average welfare and average usage of each agent (over 1 million randomly generated economies) is as shown in Figure~\ref{fig:uniform_naive_array_beta}. Assuming that agents are fully sophisticated instead, the average welfare and utilization of each agent is as shown in Figure~\ref{fig:uniform_soph_array_beta}.

\begin{figure}[t!]
\centering
\begin{subfigure}[t]{0.45\textwidth}
	\centering
	\includegraphics[scale=\figScale]{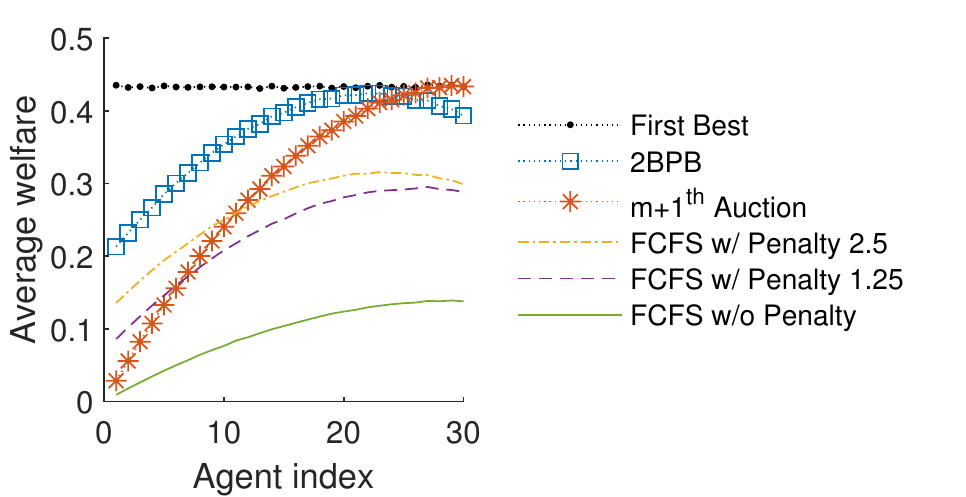}
	\caption{Average welfare. \label{fig:uniform_naive_array_beta_welfare}}
\end{subfigure}%
\hspace{1em}
\begin{subfigure}[t]{0.45\textwidth}
	\centering
 	\includegraphics[scale=\figScale]{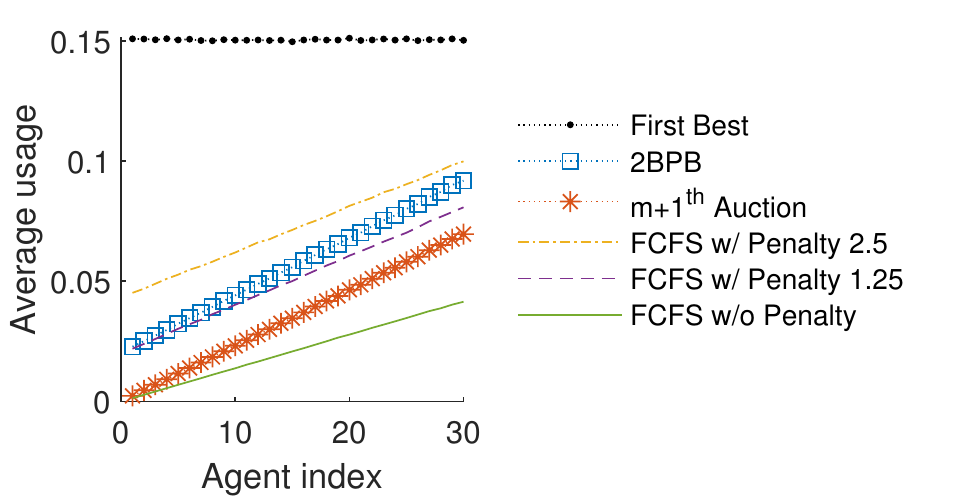}
	\caption{Average usage. \label{fig:uniform_naive_array_beta_utilization}}
\end{subfigure}%
\caption{Average welfare and usage for naive agents with uniform types, fixing $\beta_i = i/n$. 
\label{fig:uniform_naive_array_beta}
}
\end{figure}

\begin{figure}[t!]
\centering
\begin{subfigure}[t]{0.45\textwidth}
	\centering
	\includegraphics[scale=\figScale]{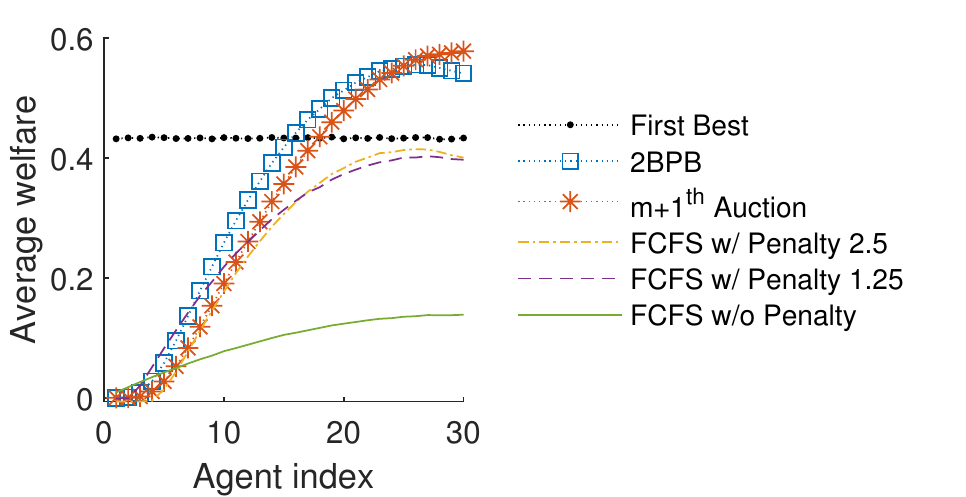}
	\caption{Social welfare. \label{fig:uniform_soph_array_beta_welfare}}
\end{subfigure}%
\hspace{1em}
\begin{subfigure}[t]{0.45\textwidth}
	\centering
 	\includegraphics[scale=\figScale]{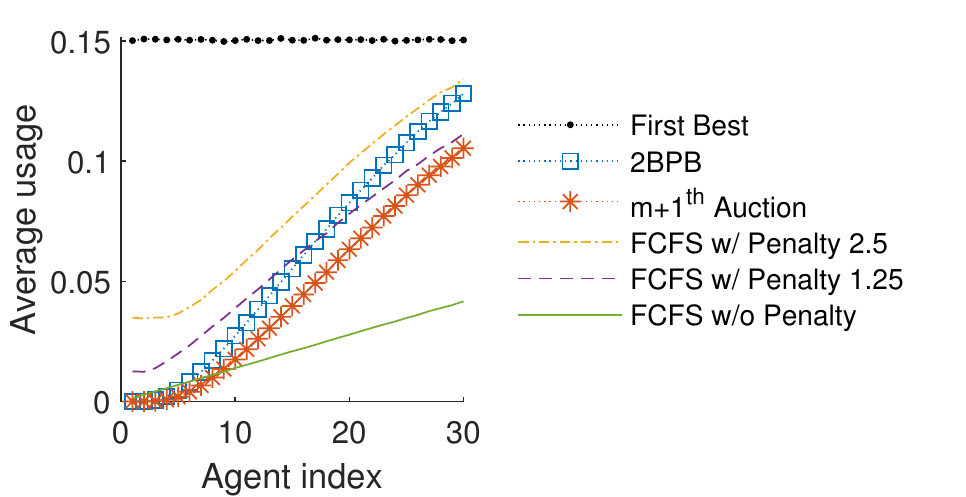}
	\caption{Average usage. \label{fig:uniform_soph_array_beta_utilization}}
\end{subfigure}%
\caption{Average welfare and usage for sophisticated agents with uniform types, fixing $\beta_i=i/n$. 
\label{fig:uniform_soph_array_beta}
}
\end{figure}

\section{Additional Examples and Discussion} \label{appx:additional_discussion}

In this section, we provide additional discussions and examples that are omitted from the body of this paper.

\subsection{Monotonicity of Bids} \label{appx:bid_monotonicity}

The following proposition shows that the less biased an agent believes she is, the higher she bids in DSE under the 2BPB mechanism and the $m+1\th$ price auction. 

\begin{proposition} \label{prop:bid_monotonicity} Under the 2BPB mechanism, or the $m+1\th$ price auction, an agent's bid in dominant strategy is monotonically increasing in her subjective present bias factor $\betahat_i$.
\end{proposition}

\begin{proof}

We first prove that for any penalty $z$, an agent's subjective expected utility $\uhat_i(z)$ is monotonically increasing in $\betahat_i$. Let $i$ and $i'$ be two agents who are identical except that $\betahat_i \geq \betahat_{i'}$. For any $z \in \setR$, we have 
\begin{align*}
	& \uhat_i(z) - \uhat_{i'}(z) \\
	= &  \E{ (V_i\1  + v_i\2) \one{V_i\1 \geq -z - \betahat_i v_i\2}} \notag  - z \Pm{V_i\1< -z - \betahat_i v_i\2} \\
	& - \E{ (V_i\1  + v_i\2) \one{V_i\1 \geq -z - \betahat_i' v_i\2}} \notag  + z \Pm{V_i\1< -z - \betahat_i' v_i\2} \\
	=& \E{ (V_i\1  + v_i\2 + z) \one{V_i\1 \in [-z - \betahat_i v_i\2, -z - \betahat_i' v_i\2)}} \\
	\geq &  0.
\end{align*}
The last inequality holds since $z \in [-z - \betahat_i v_i\2, -z - \betahat_i' v_i\2)$, $V_i\1  + v_i\2 + z \geq -z - \betahat_i v_i\2 + v_i\2 + z = (1 - \betahat_i) v_i\2 \geq 0$. 

This immediately implies the monotonicity of bids under the $m+1\th$ price auction, where agents bid $\uhat_i(0)$ in DSE. Agent $i'$ will also bid higher under the 2BPB mechanism, since the $\uhat_i(z) \geq \uhat_{i'}(z)$ for all $z \in \setR$ also implies that the zero-crossings of $\Uhat_i(z)$ and $\Uhat_{i'}(z)$ also satisfy $\zc_i \geq \zc_{i'}$.
\end{proof}

\subsection{The 2BPB Mechanism Does Not Optimize Utilization} \label{appx:not_utilization_opt}

In this section, we provide two examples which illustrate that when agents are not fully rational, the 2BPB mechanism does not necessarily optimize utilization.

The first examples shows that the 2BPB mechanism may end up achieving zero utilization and welfare by allocating to a naive agent and charging a penalty that is too small to incentivize utilization.

\begin{example} \label{exmp:two_bid_suboptimal_utilization_1}

Consider the allocation of one resource to two agents with $(c_i, p_i)$ types, where 
\begin{enumerate}[$\bullet$]
	\item $c_1 = 5$, $p_1 = 0.8$, $w_1 = 7.5$, $\beta_1 = 0.2$, $ \betahat_1 = 1$,
	\item $c_2 = 5$, $p_2 = 1/6$, $w_2 = 20$, $\beta_2 = \betahat_2 = 1$.
\end{enumerate}

Agent $1$ is fully naive and agent $2$ is fully rational. When $z < c_1 - \beta_1 w_1 = 3.5$, agent $1$ never uses the resource. 
The expected utility functions and the subjective expected utility functions of the two agents are as shown in Figures~\ref{fig:exmp_not_opt_1_u} and \ref{fig:exmp_not_opt_1_uhat}.

\begin{figure}[t!]
\centering
\begin{tikzpicture}[scale = 0.4][font = \small]

\draw[->] (-0.15,0) -- (12,0) node[anchor=north] {$z$};

\draw[->] (0,-3.5) -- (0, 3.5) node[anchor=west] {$u_i(z)$};

\draw[-] (0, 0) -- (3.45,-3.45);
\draw[-] (3.5, 1.3) -- (11,-0.2);
\draw[dotted]	(3.5,-3.5) -- (3.5, 1.3);
\draw (3.5,-3.5) circle (2.5pt);
\filldraw [black] (3.5, 1.3) circle (2.5pt);

\draw[dashed](0, 2.5) -- (3, 0) -- (4, -2.5/3);

\draw	(2.8, -0.1) node[anchor=north] {$\zc_2$}
		(10,-0.1) node[anchor=north] {$\zc_1$};
		
\draw[-] (8,2.5) -- (9,2.5) node[anchor=west] {$u_1(z)$};
\draw[dashed] (8,1.5) -- (9,1.5) node[anchor=west] {$u_2(z)$};
		
\end{tikzpicture}
\caption{The expected utility functions of two agents in
  Example~\ref{exmp:two_bid_suboptimal_utilization_1}. 
   \label{fig:exmp_not_opt_1_u}
} 
\end{figure}
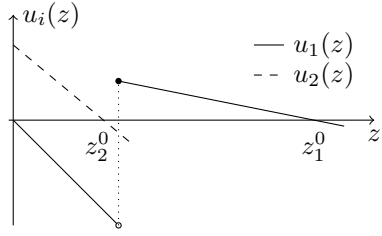

\begin{figure}[t!]
\centering
\begin{tikzpicture}[scale = 0.4][font = \small]

\draw[->] (-0.15,0) -- (12,0) node[anchor=north] {$z$};

\draw[->] (0,-0.8) -- (0, 3.5) node[anchor=west] {$\uhat_i(z)$};

\draw[-] (0, 2) -- (11,-0.2);

\draw[dashed](0, 2.5) -- (3, 0) -- (4, -2.5/3);

\draw	(2.8, -0.1) node[anchor=north] {$\zc_2$}
		(10,-0.1) node[anchor=north] {$\zc_1$};
		
\draw[-] (8,2.5) -- (9,2.5) node[anchor=west] {$\uhat_1(z)$};
\draw[dashed] (8,1.5) -- (9,1.5) node[anchor=west] {$\uhat_2(z)$};
		
\end{tikzpicture}
\caption{The subjective expected utility functions of two agents in
  Example~\ref{exmp:two_bid_suboptimal_utilization_1}. 
   \label{fig:exmp_not_opt_1_uhat}
} 
\end{figure}
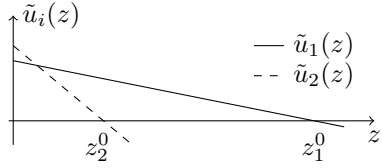

Under the second price auction, agents bid in DSE $b_{1, \txtSP}^\ast = \uhat_1(0) = 2$ and $b_{2, \txtSP}^\ast = \uhat_2(0) = (w_2 - c_2)p_2 = 2.5$. Agent $2$ gets assigned the resource and charged no penalty, achieving social welfare $(w_2-c_2)p_2 = 2.5$ and utilization $p_2 = 1/6$.

Under the 2BPB mechanism, the agents bid in DSE $\bmax_1^\ast = \zc_1 = (w_1 - c_1)p_1/(1-p_1) = 10$, and $\bmax_2^\ast = \zc_2 = (w_2 - c_2)p_2/(1-p_2) = 3$. Agent~$1$ is therefore assigned the resource, and will bid $\bmin_1^\ast = 3$ when asked to choose a penalty weakly above $\bmax_2^\ast = 3$ (this is because $\uhat_1(z)$ is monotonically decreasing in $z$ due to agent $1$'s naivete). When period $1$ comes, however, agent $1$ never shows up since the utility from using the resource appears to be $\beta_1 v_1\2 - c_1 = -3.5$, which is worse than paying the penalty and get -3. The 2BPB mechanism therefore achieves zero welfare and utilization. 
\qed
\end{example}

The second example shows that with fully sophisticated agents, it is still possible for the second price auction to achieve higher utilization.  

\begin{example} \label{exmp:two_bid_suboptimal_utilization_2}

Consider the allocation of one resource to two agents with $(c_i, p_i)$ types, where 
\begin{enumerate}[$\bullet$]
	\item $c_1 = 10$, $p_1 = 0.5$, $w_1 = 20$, $\beta_1 = \betahat_1 = 0.2$,
	\item $c_2 = 5$, $p_2 = 0.6$, $w_2 = 10$, $\beta_2 = \betahat_2 = 1$.
\end{enumerate}

Agent 1 is fully sophisticated, and agent $2$ is fully rational. $\uhat_i(z) = u_i(z)$ holds for both agents, and the expected utility functions are as shown in Figure~\ref{fig:exmp_not_opt_2_u}.

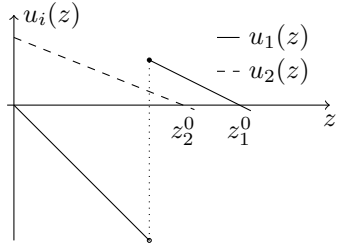
\begin{figure}[t!]
\centering
\begin{tikzpicture}[scale = 0.3][font = \small]

\draw[->] (-0.3,0) -- (14,0) node[anchor=north] {$z$};

\draw[->] (0,-6) -- (0, 4) node[anchor=west] {$u_i(z)$};

\draw[-] (0, 0) -- (6,-6);
\draw[-] (6, 2) -- (10.5, -0.25);
\draw[dotted]	(6,-6) -- (6, 2);
\draw (6,-6) circle (2.5pt);
\filldraw [black] (6, 2) circle (2.5pt);

\draw[dashed](0, 3) -- (7.5, 0) -- (8, -0.5/7.5*3);

\draw	(7.5, -0.1) node[anchor=north] {$\zc_2$}
		(10,-0.1) node[anchor=north] {$\zc_1$};
		
\draw[-] (9,3) -- (10,3) node[anchor=west] {$u_1(z)$};
\draw[dashed] (9,1.5) -- (10,1.5) node[anchor=west] {$u_2(z)$};
		
\end{tikzpicture}
\caption{The expected utility functions of two agents in
  Example~\ref{exmp:two_bid_suboptimal_utilization_2}. 
   \label{fig:exmp_not_opt_2_u}
} 
\end{figure}

Under the second price auction, agents bid in DSE $b_{1, \txtSP}^\ast = \uhat_1(0) = 0$ and $b_{2, \txtSP}^\ast = \uhat_2(0) = (w_2 - c_2)p_2 = 3$. Agent $2$ gets assigned the resource and charged no penalty, achieving social welfare $(w_2-c_2)p_2 = 3$ and utilization $p_2 = 0.6$.

Under the 2BPB mechanism, the agents bid in DSE $\bmax_1^\ast = \zc_1 = (w_1 - c_1)p_1/(1-p_1) = 10$, and $\bmax_2^\ast = \zc_2 = (w_2 - c_2)p_2/(1-p_2) = 7.5$. Agent~$1$ is therefore assigned the resource, and will bid $\bmin_1^\ast = 7.5$ when asked to choose a penalty weakly above $\bmax_2^\ast = 7.5$. Therefore, the 2BPB mechanism achieves social welfare $(w_1-c_1)p_1 = 5$ and utilization $p_1 = 0.5$. 
\qed

\end{example}

\section{Utilities and DSE Bides Under Different Type Models} \label{appx:derivations}

In this section, we derive for various type models the expected utility function and the dominant strategy equilibrium under different mechanisms.

\subsection{$(c_i, p_i)$ Type Model}

Consider an agent with $(c_i, p_i)$ type parametrized by $(c_i, p_i, w_i, \beta_i, \betahat_i,)$ who face 
a no-show penalty $z \in \setR$. In period~1, with probability $1-p_i$, the agent cannot show up, therefore gets utility $- z$. When probability $p_i$, the agent can show up at an immediate cost $c_i$. The agent believes that she will show up if and only if 
\begin{align*}
	\betahat_i w_i - c_i \geq -z_i \Leftrightarrow z_i \geq c_i - \betahat_i w_i.
\end{align*}
Therefore, $c_i - \betahat_i w_i$ is the ``minimum commitment'' the agent believes that she needs to ever show up to use the resource. When $z < c_i - \betahat_i w_i$, the agent never shows up and gets utility $-z$. When $z \geq c_i - \betahat_i w_i$, the agent does show up with probability $p_i$. The subjective expected utility of this agent is therefore:
\begin{align*}
	\uhat_i(z) = \pwfun{ -z, &\txtif z < c_i - \betahat_i w_i,\\
	(w_i - c_i)p_i - z_i (1-p_i),  & \txtif z \geq c_i - \betahat_i w_i.}
\end{align*}

When $c_i - \betahat_i w_i > 0$, the agent believes that she will not show up in a 2nd price auction, in which case she bids zero in DSE. When $c_i - \betahat_i w_i \leq 0$, she believes that she will show up with probability $p_i$, and bids her expected utility $(w_i - c_i) p_i$ from using the resource. The DSE bids under SP are therefore:
\begin{align*}
	b_{i, \txtSP}^\ast = \pwfun{(w_i - c_i) p_i, & \txtif c_i - \betahat_i w_i \leq 0, \\
	0, & \txtif c_i - \betahat_i w_i > 0.}
\end{align*}
The zero-crossing of the curve $(v_i - c_i)p_i - z_i (1-p_i)$ is 
\begin{align*}
	\zcmaxhat_i = \frac{(w_i - c_i) p_i}{1-p_i},
\end{align*}
therefore when $\zcmaxhat_i < c_i - \betahat_i w_i$, $\uhat_i(z) < 0$ for any $z > 0$, meaning that the agent will not participate in the 2BPB mechanism. 
When $\zcmaxhat_i \geq c_i - \betahat_i w_i$, we know that $\zc_i = \zcmaxhat_i$ is the zero-crossing of $\Uhat_i(z)$ , therefore the DSE bid on the maximum acceptable penalty is
\begin{align*}
	\bmax_{i}^\ast =& \pwfun{\zcmaxhat_i, & \txtif c_i - \betahat_i w_i \leq \zcmaxhat_i, \\
	0, & \txtif c_i - \betahat_i w_i > \zcmaxhat_i.}
\end{align*}

We also know that $\uhat_i(z)$ is monotonically decreasing when $z \geq  c_i - \betahat_i w_i$, therefore after given a minimum penalty $\underline{z}$, the agent will bid in DSE her preferred penalty 
\begin{align*}
	\bmin_{i}^\ast = \max\{c_i - \betahat_i w_i, ~\underline{z} \}.
\end{align*}

If agent $i$ is assigned a resource and charged a penalty $z$, the actual utilization would be 
\begin{align*}
	\ut_i(z) =  p_i \cdot \one{z \geq c_i - \beta_i w_i},
\end{align*}
since when time $1$ comes, she will discount the future utility according to her true discounting factor $\beta_i$. The  expected social welfare is therefore
\begin{align*}
	\sw_i(z) =  p_i(w_i - c_i) \cdot \one{z \geq c_i - \beta_i w_i},
\end{align*}
and the agent's actual expected utility is
\begin{align*}
	u_i(z) = \pwfun{ -z, &\txtif z < c_i - \beta_i w_i,\\
	(v_i - c_i)p_i - z_i (1-p_i),  & \txtif z \geq c_i - \beta_i w_i.}
\end{align*}

The first-best utilization that can be achieved by this agent is
\begin{align*}
	\ut_i\fb =  p_i,
\end{align*}
and the first-best welfare is:
\begin{align*}
	\sw_i\fb = (w_i - c_i) p_i \one{c_i - \betahat_i w_i \leq \zcmaxhat_i}. 
\end{align*}
Note that even when $\zcmaxhat_i < c_i - \betahat_i w_i$ in which case $\uhat_i(z) < 0$ for any $z > 0$, we may still incentivize the agent to show up with probability $p_i$ by charging a no-show penalty weakly higher than $ \betahat_i w_i$, and also making a positive payment to the agent to incentivize participation. It is possible to do this without running a deficit since we achieve a positive welfare $(w_i - c_i)p_i$.

\subsection{Exponential Type Model}

Consider now the exponential type model, where an agent's type is parametrized by $\theta_i = (\lambda_i, w_i, \beta_i, \betahat_i)$. In period~$1$, with penalty $z$, the agent will show up to use the resource if and only if 
\begin{align*}
	V_i\1 + \beta_i w_i \geq -z \Leftrightarrow V_i\1 \geq -z - \beta_i w_i. 
\end{align*}
This happens with probability $\Pm{V_i\1 \geq -z - \beta_i w_i} = 1- e^{-\lambda_i(z + \beta_i w_i)}$, as long as $z + \beta_i w_i \geq 0 \Leftrightarrow z \geq - \beta_i w_i$. Therefore the actual utilization as a function of penalty $z$ is:
\begin{align*}
	\ut_i(z) = \pwfun{ 1- e^{-\lambda_i(z + \beta_i w_i)}, & \txtif z \geq - \beta_i w_i\\
	0, & \txtif z < - \beta_i w_i,}
\end{align*}
and the expected social welfare is:
\begin{align*}
	& \sw_i(z) \\ = & \pwfun{ w_i - \frac{1}{\lambda_i} + 
 e^{-\lambda_i (\beta_i w_i + z)} \left( \frac{1}{\lambda_i} - (1 - \beta_i) w_i + z) \right), & \txtif z \geq - \beta_i w_i\\
	0, & \txtif z < - \beta_i w_i.}
\end{align*}
With $z < -\beta_i w_i$, the agent never uses the resource, and gets expected utility $u_i(z) = -z$. When $z \geq -\beta_i w_i$, the agent gets expected utility:
\begin{align*}
	& \ut_i(z) \\ = & \E{ (V_i\1 + w_i) \one{V_i\1 + \beta_i w_i \geq -z} } - z \Pm{V_i\1 + \beta_i w_i <-z} \\
	= & \int_{0}^{z + \beta_i w_i} (-v + w_i) \lambda_i e^{-\lambda_i v} dv - z e^{-\lambda_i(\beta_i w_i + z)} \\ 
	=& w_i - 1/\lambda_i + e^{-\lambda_i (\beta_i w_i + z)} (1/\lambda_i - (1 - \beta_i)  w_i) .
\end{align*}
The agent, however, believes that her present bias factor is $\betahat_i$, therefore believes that her expected utility as a function of the penalty $z$ is:
\begin{align*}
	& \uhat_i(z) \\  = & \pwfun{ w_i - 1/\lambda_i + e^{-\lambda_i (\betahat_i w_i + z)} (1/\lambda_i - (1 - \betahat_i)  w_i), & \txtif z \geq -\betahat_i w_i, \\ 
	-z, & \txtif z < -\beta_i w_i.}
\end{align*}
%
$\uhat_i(0) \geq 0$ always holds, therefore under SP, the agent is going to bid:
\begin{align*}
	b_{i,\txtSP}^\ast = \uhat_i(0) =  w_i - 1/\lambda_i + e^{-\lambda_i \betahat_i w_i } (1/\lambda_i - (1 - \betahat_i)  w_i).
\end{align*}
Taking the derivative of $\uhat_i(z)$ w.r.t. $z$ for $z \geq -\betahat_i w_i$, we have:
\begin{align*}
	\frac{d}{dz}\uhat_i(z) = e^{-\lambda_i (\betahat_i w_i + z)} (-1 + \lambda_i w_i(1 - \betahat_i)). 
\end{align*}
When (A3) holds, $w_i < 1/\lambda_i$ implies $\frac{d}{dz}\uhat_i(z) < 0$, meaning that $\uhat_i(z)$ is monotonically decreasing in $z$, and that $\Uhat_i(z)$ and $\uhat_i(z)$ coincide. The zero-crossing (i.e. the maximum acceptable penalty) is therefore equal to 
\begin{align*}
	\zc_i = -\betahat_i w_i + \frac{1}{\lambda_i} \log \left( \frac{1 - \lambda_i w_i (1 -  \betahat_i)}{1 - \lambda_i w_i }\right).
\end{align*}
Under the 2BPB mechanism, the agent is going to bid in DSE a maximum penalty
\begin{align*}
	\bmax_{i}^\ast = \zc_i,
\end{align*}
and once given a minimum penalty $\underline{z}$, the agent will then bid the smallest possible $\bmin_{i}^\ast = \underline{z}$. 
The first-best social welfare for this agent can be achieved by setting $z = (1-\beta_i)w_i$, in which case the agent will use the resource if and only if $V_i\1 + w_i \geq 0$. The first-best welfare is therefore:
\begin{align*}
	\sw_i\fb =  w_i + ( e^{-\lambda_i w_i} - 1)/\lambda_i.
\end{align*}
The first-best utilization is achieved by charging the highest penalty s.t. $\sw_i(z) \geq 0$ still holds, i.e. it is possible for the outcome to be both budget balanced and individually rational. Solving the equation, we get the maximum penalty that we can charge as:
\begin{align*}
	z_i\fb = -1/\lambda_i + (1-\beta_i) w_i + \frac{1}{\lambda_i}
 \mathrm{ProductLog} \left(-1, e^{-1 + \lambda_i w_i} (-1 + \lambda_i w_i)\right).
\end{align*}
Here, $\mathrm{ProductLog}$ (also called the Lambert $W$ function) is the inverse relation of the function $f(s) = se^s$. 
The first best utilization achieved at penalty $z_i\fb $ is therefore:
\begin{align*}
	\ut_i\fb = \ut_i(z_i\fb) = 1 - e^{1 - \lambda_i w_i + 
  \mathrm{ProductLog}\left(-1, e^{-1 + \lambda_i w_i} (-1 + \lambda_i w_i) \right)}. 
\end{align*}

\subsection{Uniform Type Model}

We now consider the uniform type model, where an agent type is parametrized by $(\alpha_i, w_i, \beta_i, \betahat_i)$. In period~1, with penalty $z$, the agent will show up to use the resource if and only if 
\begin{align*}
	V_i\1 + \beta_i w_i \geq -z \Leftrightarrow V_i\1 \geq -z - \beta_i w_i. 
\end{align*}
With $V_i\1 \sim \mathrm{U}[-\alpha_i, ~0]$, we know that there are three cases depending on $z$:
\begin{enumerate}[$\bullet$]
	\item when $z \leq -\beta_i w_i$, $-z - \beta_i w_i$ is strictly positive, thus the agent never shows up, resulting in utilization and welfare both equal to zero.
	\item when $z > \alpha_i -\beta_i w_i$, $-z - \beta_i w_i < -\alpha_i$ so that the agent always shows up. The utilization is therefore equal to $1$, and the welfare is equal to $\E{V_i\1 + w_i} = w_i -\alpha_i/2$. 
	\item when $z \in (-\beta_i w_i, \alpha_i -\beta_i w_i]$, the agent shows up with probability $(z + \beta_i w_i)/\alpha_i$. 
\end{enumerate}
Putting the three cases together, we know that the utilization as a function of the penalty $z$ is:
\begin{align*}
	\ut_i(z) = \pwfun{0, &\txtif z < -\beta_i w_i, \\
		(z + \beta_i w_i)/\alpha_i, & \txtif -\beta_i w_i \leq z < \alpha_i - \beta_i w_i, \\
		1, & \txtif z > \alpha_i - \beta_i w_i.} 
\end{align*}
The expected social welfare is:
\begin{align*}
	\sw_i(z) = \pwfun{0, &\txtif z < -\beta_i w_i, \\
		\frac{z + \beta_i w_i}{\alpha_i} \left( w_i - \frac{z + \beta_i w_i}{2}\right), & \txtif -\beta_i w_i \leq z < \alpha_i - \beta_i w_i, \\
		w_i - \alpha_i/2, & \txtif z > \alpha_i - \beta_i w_i,}
\end{align*}
and the agent's expected utility is:
\begin{align*}
	& u_i(z)  \\ = & \pwfun{-z, &\txtif z < -\beta_i w_i, \\
		\frac{z + \beta_i w_i}{\alpha_i} \left( w_i - \frac{z + \beta_i w_i}{2}\right) - z \frac{\alpha_i - (z + \beta_i w_i)}{\alpha_i}, & \txtif -\beta_i w_i \leq z < \alpha_i - \beta_i w_i, \\
		w_i - \alpha_i/2, & \txtif z > \alpha_i - \beta_i w_i.}
\end{align*}
$\uhat_i(z)$ can be obtained simply by replacing $\beta_i$ with $\betahat_i$ in the above expression. 
$\uhat_i(0) \geq 0$ always holds, Therefore under SP, the agent is going to bid:
\begin{align*}
	b_{i,\txtSP}^\ast = \uhat_i(0) =  \frac{\betahat_i w_i}{\alpha_i} \left( w_i - \frac{\betahat_i w_i}{2}\right).
\end{align*}
Note that for $-\betahat_i w_i \leq z < \alpha_i - \betahat_i w_i$, $\uhat_i(z)$ can be rewritten in the following quadratic form: 
\begin{align*}
	\uhat_i(z) = \frac{1}{2\alpha_i}\left( z^2 - 2(\alpha_i - w_i)z + w_i^2 \betahat_i (2 - \betahat_i)\right).
\end{align*}
The minimum is achieved at $z_i^\ast = \alpha_i - w_i$, which is 
\begin{align*}
	\uhat_i(z_i^\ast) =  \frac{1}{2\alpha_i}\left( - (\alpha_i - w_i)^2 + w_i^2 \betahat_i (2 - \betahat_i)\right).
\end{align*}
When (A3) holds i.e. $\E{V_i\1 + w_i} < 0$, we have $w_i < \alpha_i/2$, which implies $\uhat_i(z^\ast) \leq 0$ for any $\betahat_i \in [0,1]$. As a result, $\uhat_i(z)$ is monotonically decreasing in $z$ for $z \leq z_i^*$, monotonically increasing for $z > z_i^*$, and $\uhat_i(z) \leq 0$ holds for all $z < z_i^*$. This implies that 
%
$\uhat_i(z) $ and $\Uhat_i(z)$ coincide for all $z$ s.t. $\uhat_i(z) \geq 0$, and that the zero-crossing of $\uhat_i(z)$ and $\Uhat_i(z)$ (i.e. the maximum acceptable penalty) is of the form:
\begin{align*}
	\zc_i = \alpha_i - w_i - \sqrt{\alpha_i^2 - 2 \alpha_i w_i + (-1 + \betahat_i)^2 w_i^2}.
\end{align*}
The DSE bid on maximum penalty under the 2BPB mechanism is therefore:
\begin{align*}
	\bmax_{i}^\ast = \zc_i,
\end{align*}
and once allocated given a minimum penalty $\underline{z}$, the agent will then bid $\bmin_{i}^\ast = \underline{z}$. 
Similar to the exponential model, the first-best welfare is achieved by setting the penalty as $z = (1-\beta_i)w_i$, in which case
\begin{align*}
	\sw_i\fb = \frac{w_i}{\alpha_i} \left( w_i - \frac{w_i}{2}\right) = \frac{w_i^2}{2\alpha}.
\end{align*}
The first-best social welfare is achieved at $z = 2 w_i - \beta_i w_i$, in which case:
\begin{align*}
	\ut_i\fb = (2w_i - \beta_i w_i + \beta_i w_i) / \alpha_i = 2w_i/\alpha_i. 
\end{align*}

\end{document}